\documentclass[11pt]{article}

\usepackage{graphicx}
\usepackage{epsfig}
\usepackage{epsf}
\usepackage{latexsym}
\usepackage{bbm}
\usepackage{float}
\usepackage{fancyvrb}
\usepackage[small]{caption}
\usepackage{url}

\usepackage{amsthm}
\usepackage{amsfonts}
\usepackage{amssymb}
\usepackage{amsmath}
\usepackage{enumerate}

\newtheorem{theorem}{Theorem}
\newtheorem{lemma}{Lemma}
\newtheorem{proposition}{Proposition}

\newtheorem{observation}{Observation}

\newtheorem{algorithm}{Algorithm}

\newcommand{\opt}{{\rm OPT}}

\newcommand{\area}{{\rm area}}
\newcommand{\inter}{{\rm int}}

\newcommand{\NN}{\mathbb{N}} 
\newcommand{\RR}{\mathbb{R}} 

\newcommand{\eps}{\varepsilon}

\def\Q{\mathcal Q}
\def\F{\mathcal F}

\def\eg{{e.g.}}
\def\ie{{i.e.}}

\newcommand{\later}[1]{{}}
\newcommand{\old}[1]{{}}
\long\def\ignore#1{}

\title{Anchored Rectangle and Square Packings\thanks{This work was
    partially supported by the NSF awards CCF-1422311 and CCF-1423615.}}

\author{%
Kevin Balas\thanks{Department of Mathematics, California State
  University Northridge, Los Angeles, CA, USA\@.}\ \thanks{Mathematics
  Department, Los Angeles Mission College, Sylmar, CA, USA\@. Email:
  \texttt{balask@lamission.edu}.}
\qquad
Adrian Dumitrescu\footnote{Department of Computer Science, University
  of Wisconsin--Milwaukee, WI, USA\@.
Email: \texttt{dumitres@uwm.edu}.}
\qquad
Csaba D. T\'oth\footnotemark[1]\ \thanks{Department of Computer Science,
Tufts University, Medford, MA, USA\@. Email:~\texttt{cdtoth@acm.org}.
}}

\begin{document}

\maketitle

\begin{abstract}
For points $p_1,\ldots , p_n$ in the unit square $[0,1]^2$, an
\emph{anchored rectangle packing} consists of interior-disjoint
axis-aligned empty rectangles $r_1,\ldots , r_n\subseteq [0,1]^2$ such
that point $p_i$ is a corner of the rectangle $r_i$ (that is, $r_i$ is
\emph{anchored} at $p_i$) for $i=1,\ldots, n$. We show that for every
set of $n$ points in $[0,1]^2$, there is an anchored rectangle packing
of area at least $7/12-O(1/n)$, and for every $n\in \NN$, there are
point sets for which the area of every anchored rectangle packing is
at most $2/3$. The maximum area of an anchored \emph{square} packing
is always at least $5/32$ and sometimes at most $7/27$.

The above constructive lower bounds immediately yield constant-factor
approximations, of $7/12 -\eps$ for rectangles and $5/32$ for squares,
for computing anchored packings of maximum area in $O(n\log n)$ time.
We prove that a simple greedy strategy achieves a $9/47$-approximation
for anchored square packings, and $1/3$ for lower-left anchored square packings.
Reductions to maximum weight independent set (MWIS) yield a QPTAS
and a PTAS for anchored rectangle and square packings in $n^{O(1/\eps)}$
and $\exp({\rm poly}(\log (n/\eps)))$ time, respectively.


\end{abstract}

\section{Introduction} \label{sec:intro}

Let $P=\{p_1,\ldots , p_n\}$ be a finite set of points in an
axis-aligned bounding rectangle $U$.
An \emph{anchored rectangle packing} for $P$ is a set of axis-aligned empty rectangles
$r_1,\ldots , r_n$ that lie in $U$, are interior-disjoint, and $p_i$ is one of the
four corners of $r_i$ for $i=1,\ldots , n$;
rectangle $r_i$ is said to be \emph{anchored} at $p_i$.
For a given point set $P\subset U$, we wish to find the maximum total
area $A(P)$ of an anchored rectangle packing of $P$. Since the ratio
between areas is an affine invariant, we may assume that $U=[0,1]^2$.
However, if we are interested in the maximum area of an \emph{anchored
  square packing}, we \emph{must} assume that $U=[0,1]^2$ (or that the
aspect ratio of $U$ is bounded from below by a constant; otherwise,
with an \emph{arbitrary} rectangle $U$, the guaranteed area is zero).

Finding the maximum area of an anchored rectangle packing
of $n$ given points is suspected but not known to be NP-hard.
Balas and T\'oth~\cite{BT15} observed that the number of distinct
rectangle packings that attain the maximum area, $A(P)$, can be
exponential in $n$. From the opposite direction, the same authors~\cite{BT15}
proved an exponential upper bound on the number of maximum area configurations,
namely $2^n C_n= \Theta(8^n/{n^{3/2}})$,
where $C_n = \frac{1}{n+1} {2n \choose n} = \Theta(4^n/{n^{3/2}})$
is the $n$th Catalan number. Note that a greedy strategy may fail to
find $A(P)$; see Fig.~\ref{fig:greedyfail}. 

\begin{figure}[hptb]
\centering
\includegraphics[scale=1.2]{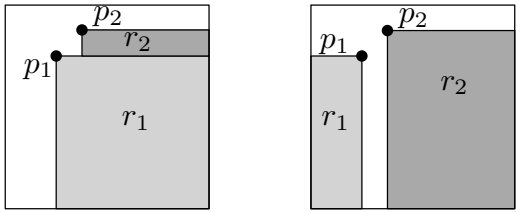}
\caption{For $P=\{p_1,p_2\}$, with $p_1 =(\frac{1}{4},\frac{3}{4})$
  and $p_2 = (\frac{3}{8},\frac{7}{8})$,
a greedy algorithm selects rectangles of area
$\frac{3}{4}\cdot \frac{3}{4} + \frac{1}{8}\cdot \frac{5}{8} = \frac{41}{64}$ (left),
which is less than the area
$\frac{1}{4}\cdot \frac{3}{4} + \frac{5}{8}\cdot \frac{7}{8} = \frac{47}{64}$
of the packing on the right.}
\label{fig:greedyfail}
\end{figure}

\paragraph{Variants and generalizations.}
We consider three additional variants of the problem.
An \emph{anchored square packing} is an anchored rectangle packing
in which all rectangles are squares; a \emph{lower-left anchored rectangle packing}
is a rectangle packing where each point $p_i\in r_i$ is the lower-left corner of $r_i$;
and a \emph{lower-left anchored square packing} has both properties.

We suspect that all
variants, with rectangles or with squares, are NP-hard.
Here, we put forward several approximation algorithms, while it is understood that
the news regarding NP-hardness can occur at any time
or perhaps take some time to establish.

The problem can be generalized to other geometric shapes with distinct representatives.
Let $P=\{p_1,\ldots,p_n\}$ be a finite set of points in a compact domain $U\subset \RR^d$,
and let $\F=\{\F_1,\ldots,\F_n\}$ be $n$ families of measurable sets
(\eg, rectangles, squares, or disks) such that for all $r\in \F_i$,
we have $p_i\in r\subseteq U$ and $\mu(r)\geq 0$ is the measure of $r$.
An \emph{anchored packing for $(P,\F)$} is a set of pairwise interior-disjoint
representatives $r_1,\ldots , r_n$ with $r_i\in \F_i$ for $i=1,\ldots, n$.
We wish to find an anchored packing for $(P,\F)$ of maximum measure $\sum_{i=1}^n\mu(r_i)$.
While some variants are trivial (\eg, when $U=[0,1]^2$ and $\F_i$ consists of all rectangles
containing $p_i$), there are many interesting and challenging variants
(\eg, when $\F_i$ consists of disks containing $p_i$; or when $U$ is nonconvex).
In this paper we assume that the domain $U$ and the families $\F$ are
axis-aligned rectangles or axis-aligned squares in the plane.

\begin{table}[hptb]
\begin{center}
 \begin{tabular}{|l||c|c|c|c|}
   \hline 
Anchored packing with & rectangles & squares &  LL-rect.  & LL-sq. \\ \hline
\rule{0pt}{2.5ex}Guaranteed max.~area & $\frac{7}{12}-O(\frac{1}{n})\leq A(n)\leq
\frac{2}{3}$ & $\frac{5}{32}\leq A_\textup{sq}(n) \leq \frac{7}{27}$ &
$0$ & $0$\\
Greedy approx.~ratio  & $7/12-\eps$  & $9/47$ & $0.091$~\cite{DT15}  & $1/3$\\
Approximation scheme        & QPTAS & PTAS & QPTAS & PTAS\\
   \hline
 \end{tabular}
\caption{Table of results for the four variants studied in this
  paper. The last two columns refer to lower-left anchored rectangles
  and lower-left anchored squares, respectively.\label{table:1}}
\end{center}
\end{table}

\vspace{-2\baselineskip}
\paragraph{Contributions.} Our results are summarized in Table~\ref{table:1}.

(i) We first deduce upper and lower bounds on the maximum area of
an anchored rectangle packing of $n$ points in $[0,1]^2$. For $n\in
\NN$, let $A(n)=\inf_{|P|=n} A(P)$. We prove that
$\frac{7}{12}-O(1/n)\leq A(n)\leq \frac{2}{3}$ for all $n\in \NN$
(Sections~\ref{sec:upper} and~\ref{sec:lower}).

\smallskip
(ii) Let $A_\textup{sq}(P)$ be the maximum area of an anchored square
packing for a point set $P$, and $A_\textup{sq}(n)=\inf_{|P|=n} A_\textup{sq}(P)$.
We prove that $\frac{5}{32}\leq A_\textup{sq}(n) \leq \frac{7}{27}$ for all $n$
(Sections~\ref{sec:upper} and~\ref{sec:lower:sq}).

\smallskip
(iii) The above constructive lower bounds immediately yield
constant-factor approximations for computing anchored packings of maximum area
($7/12 -\eps$ for rectangles and $5/32$ for squares) in $O(n\log n)$
time (Sections~\ref{sec:lower} and \ref{sec:lower:sq}). 
In Section~\ref{sec:approx} we show that a (natural) greedy algorithm for
anchored square packings achieves a better approximation ratio, namely
$9/47=1/5.22\ldots$, in $O(n^2)$ time.
By refining some of the tools developed for this bound, in Section~\ref{sec:approx:ll}
we prove a tight bound of $1/3$ for the approximation ratio of a greedy algorithm for
lower-left anchored square packings.

\smallskip
(iv) We obtain a polynomial-time approximation scheme (PTAS)
for the maximum area anchored square packing problem,
and a quasi-polynomial-time approximation scheme (QPTAS)
for the maximum area anchored rectangle packing problem,
via a reduction to the maximum weight independent set (MWIS)
problem for axis-aligned squares~\cite{Cha03} and rectangles~\cite{AW13},
respectively. Given $n$ points, an $(1-\eps)$-approximation for the anchored square packing
of maximum area can be computed in time $n^{O(1/\eps)}$; and for the
anchored rectangle packing  of maximum area, in time $\exp({\rm poly}(\log (n/\eps)))$.
Both results extend to the lower-left anchored variants (Section~\ref{sec:schemes}).

\paragraph{Motivation and related work.}
Packing axis-aligned rectangles in a rectangular container, albeit
without anchors, is the unifying theme of several classic optimization
problems. The 2D knapsack problem, strip packing, and 2D bin packing
involve arranging a set of given rectangles in the most economic
fashion~\cite{AW13,BK14}. The maximum area independent set (MAIS) problem for
rectangles (squares, or disks, etc.),  is that of selecting a maximum area
packing from a given set~\cite{AW15}; see classic papers
such as~\cite{Aj73,Ra49,Ra51,Ra68,Ra28} and also
more recent ones~\cite{BDJ10a,BDJ10b,DJ13b} for
quantitative bounds and constant approximations.
These optimization problems are NP-hard,
and there is a rich literature on approximation algorithms.
Given an axis-parallel rectangle $U$ in the plane containing $n$ points,
the problem of computing a maximum-area empty axis-parallel
sub-rectangle contained in $U$ is one of the oldest problems studied
in computational geometry~\cite{AS87,CDL86}; the higher dimensional
variant has been also studied~\cite{DJ13a}.
In contrast, our problem setup is fundamentally different: the rectangles (one for
each anchor) have variable sizes, but their location is constrained by the anchors.

Map labeling problems in geographic information systems (GIS)~\cite{KT13,KR92,KSW99}
call for choosing interior-disjoint rectangles that are incident to a given set
of points in the plane. GIS applications often impose 
constraints on the label boxes, such as aspect ratio, minimum and maximum size,
or priority weights.
Most optimization problems of such variants are known
to be NP-hard~\cite{FW91,IL03,JC04,KNN+02}.
In this paper, we focus on maximizing the total area of an anchored rectangle packing.

In a restricted setting where each point $p_i$ is the lower-left corner
of the rectangle $r_i$ and $(0,0)\in P$,
Allen Freedman~\cite{Tu69,Wi07} conjectured almost 50 years ago that
there is a lower-left anchored rectangle packing of area at least $1/2$.
The current best lower bound on the area under these conditions
is (about) $0.091$, as established in~\cite{DT15}.
The analogous problem of estimating the total area for lower-left
anchored square packings is much easier.
If $P$ consists of the $n$ points $(i/n,i/n)$, $i=0,1,\ldots,n-1$,
then the total area of the $n$ anchored squares is at most $1/n$, and so it
tends to zero as $n$ tends to infinity.
A looser anchor restriction, often appearing in map labeling problems with square labels,
requires the anchors to be contained in the boundaries of the squares,
however the squares need to be congruent; see~\eg,~\cite{ZJ06}.

In the context of \emph{covering} (as opposed to \emph{packing}),
the problem of covering a given polygon by disks of given centers
and varying radii such that the sum of areas of the disks is
minimized has been considered in~\cite{ACK+11,BCK+05}.
In particular, covering $[0,1]^2$ with $\ell_\infty$-disks of given centers
and minimal area as in~\cite{BVX13,BVX13b} is dual
to the anchored square packings studied in this paper.

\paragraph{Notation.}
Given an $n$-element point set $P$ contained in $U=[0,1]^2$,
denote by $\opt=\opt(P)$ a packing (of rectangles or squares, as the
case may be) of maximum total area.  An algorithm for a packing
problem has approximation ratio $\alpha$ if the packing it computes,
$\Pi$, satisfies $\area(\Pi) \geq \alpha \, \area(\opt)$, for some $\alpha \leq 1$.
A set of points is \emph{in general position} if no two points have the
same $x$- or $y$-coordinate. The boundary of a planar body $B$
is denoted by $\partial B$, and its interior by $\inter(B)$.

\section{Upper Bounds} \label{sec:upper}

\begin{proposition} \label{prop:2/3}
For every $n\in \NN$, there exists a point set $P_n$ such that every
anchored rectangle packing for $P_n$ has area at most $\frac{2}{3}$.
Consequently, $A(n) \leq \frac{2}{3}$.
\end{proposition}
\begin{proof}
Consider the point set $P = \{p_1,\ldots ,p_n\}$, where
$p_i = (x_i,y_i) =(2^{-i},2^{-i})$, for $i=1,\ldots,n$;
see Fig.~\ref{f3}\,(left).
Let $R = \{r_1,\ldots ,r_n\}$ be an anchored rectangle packing for $P$.
Since $p_1 = (\frac{1}{2},\frac{1}{2})$, any rectangle anchored at
$p_1$ has height at most $\frac{1}{2}$, width at most
$\frac{1}{2}$, and hence area at most $\frac{1}{4}$.

For $i=2,\ldots , n$, the $x$-coordinate of $p_i$, $x_i$, is halfway
between 0 and $x_{i-1}$, and $y_i$ is halfway between 0 and
$y_{i-1}$. Consequently, if $p_i$ is the lower-right, lower-left or
upper-left corner of $r_i$, then $\area(r_i) \leq
(\frac{1}{2^i})(1- \frac{1}{2^i})= \frac{1}{2^i} - \frac{1}{4^i}$. If, $p_i$
is the upper-right corner of $r_i$, then $\area(r_i) \leq
\frac{1}{4^i}$. Therefore, in all cases, we have $\area(r_i) \leq
\frac{1}{2^i} - \frac{1}{4^i}$. The total area of an anchored rectangle packing
is bounded from above as follows:
\begin{equation*}
A(P) \leq \sum_{i=1}^{n} \left(\frac{1}{2^i} - \frac{1}{4^i} \right)
=\left (1 - \frac{1}{2^n}\right ) - \frac{1}{3} \left (1 - \frac{1}{4^n} \right )
= \frac{2}{3} - \frac{1}{2^n} + \frac{1}{3 \cdot 4^n} \leq \frac{2}{3}.
\tag*{\qedhere}
\end{equation*}
\end{proof}

\begin{figure}[htbp]
\centerline{\epsfxsize=3.3in \epsffile{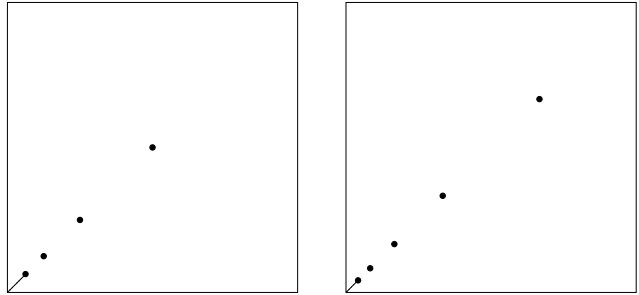}}
\caption{Left: $2/3$ upper bound construction for anchored rectangles.
Right: $7/27$ upper bound construction for anchored squares.}
\label{f3}
\end{figure}

\begin{proposition} \label{prop:7/27}
For every $n\in \NN$, there exists a point set $P_n$ such that every
anchored square packing for $P_n$ has area at most $\frac{7}{27}$.
Consequently, $A_\textup{sq}(n) \leq \frac{7}{27}$.
\end{proposition}
\begin{proof}
Consider the point set $P = \{p_1,\ldots ,p_n\}$, where $p_i =
(\frac43\cdot 2^{-i},\frac43\cdot 2^{-i})$, for $i=1,\ldots ,n$;
see Fig.~\ref{f3}\,(right).
Let $S = \{s_1,\ldots ,s_n\}$ be an anchored square packing for $P$. Since
$p_1=(\frac23,\frac23)$ and $p_2=(\frac13,\frac13)$, any square
anchored at $p_1$ or at $p_2$ has side-length at most $\frac{1}{3}$,
hence area at most $\frac{1}{9}$. For $i=3,\ldots , n$, the
$x$-coordinate of $p_i$, $x_i$, is halfway between 0 and $x_{i-1}$,
and $y_i$ is halfway between 0 and $y_{i-1}$. Hence any square
anchored at $p_i$ has side-length at most $x_i=y_i=\frac{4}{3\cdot
  2^i}$, hence area at most $\frac{16}{9\cdot 4^i}$. The total area of
an anchored square packing is bounded from above as follows:
\begin{equation*}
A_\textup{sq}(P) \leq \frac29+\frac19\sum_{j=1}^{n-1}\frac{1}{4^j}
< \frac29+\frac19\sum_{j=1}^{\infty}\frac{1}{4^j}
= \frac29 + \frac19 \cdot \frac13 = \frac{7}{27}.
\tag*{\qedhere}
\end{equation*}
\end{proof}

\paragraph{Remark.}
Stronger upper bounds hold for small $n$, \eg, $n\in \{1,2\}$.
Specifically, $A(1)=A_\textup{sq}(1)=1/4$
attained for the center $(\frac12,\frac12)\in [0,1]^2$, and $A(2)=4/9$ and
$A_\textup{sq}(2)=2/9$ attained for $P=\{(\frac13,\frac13),\frac23,\frac23)\}$.

\section{Lower Bound for Anchored Rectangle Packings} \label{sec:lower}

In this section, we prove that for every set $P$ of $n$ points in $[0,1]^2$,
we have $A(P)\geq \frac{7n-2}{12(n+1)}$. Our proof is constructive;
we give a divide \& conquer algorithm that partitions $U$ into horizontal
strips and finds $n$ anchored rectangles of total area bounded
from below as required. We start with a weaker lower bound, of about
$1/2$, and then sharpen the argument to establish the main result of
this section, a lower bound of about $7/12$.

\begin{proposition}\label{Pn/2(n+1)}
For every set of $n$ points in the unit square $[0,1]^2$,
an anchored rectangle packing of area at least $\frac{n}{2(n+1)}$
can be computed in $O(n\log n)$ time.
\end{proposition}
\begin{proof}
Let $P = \{p_1,\ldots ,p_n\}$ be a set of points in the unit square
$[0,1]^2$ sorted by their $y$-coordinates. Draw a horizontal line through each point in $P$;
see Fig.~\ref{fig:horizontalstrips}\,(left). These lines divide $[0,1]^2$ into $n+1$
horizontal strips.  A strip can have zero width if two points have the
same $x$-coordinate.  We leave a narrowest strip empty and assign the
remaining strips to the $n$ points such that each rectangle above
(resp., below) the chosen narrowest strip is assigned to a point of
$P$ on its bottom (resp., top) edge. For each point divide the
corresponding strip into two rectangles with a vertical line through
the point. Assign the larger of the two rectangles to the point.

\begin{figure}[hptb]
\centering
\includegraphics[scale=1]{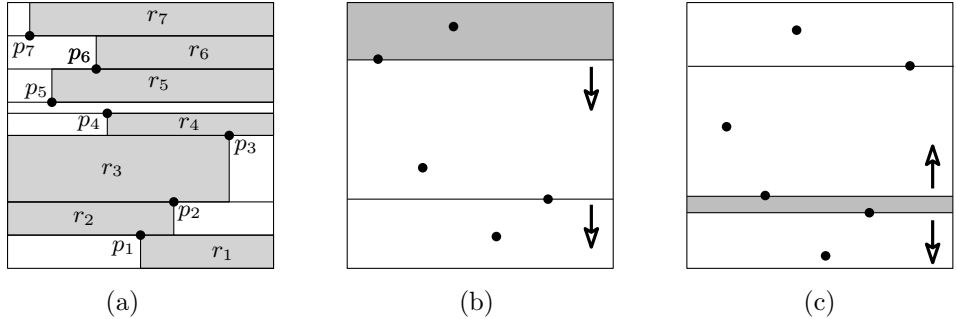}
\caption{Left: Horizontal strips with anchored rectangles for 7 points.
Middle: an example of the partition for odd $n$ (here $n=5$).
Right: an example of the partition for even $n$ (here $n=6$).
The strip that is discarded is shaded in the drawing.}
\label{fig:horizontalstrips}
\end{figure}

The area of narrowest strip is at most $\frac{1}{n+1}$. The rectangle
in each of the remaining $n$ strips covers at least $\frac{1}{2}$ of
the strip. This yields a total area of at least $\frac{n}{2(n+1)}$.
\end{proof}

For a stronger lower bound, our key observation is that for \emph{two}
points in a horizontal strip, one can always pack two anchored
rectangles in the strip that cover strictly more than half the area of
the strip. Specifically, we prove the following easy-looking statement
with $2$ points in a rectangle (however, we did not find an easy proof!).
The proof of Lemma~\ref{lem:rectangle} is deferred to Appendix~\ref{sec:rectangle}.

\begin{lemma} \label{lem:rectangle}
Let $P=\{p_1,p_2\}$ be two points in an axis-parallel rectangle $R$
such that $p_1$ lies on the bottom side of $R$.
Then there exist two empty rectangles in $R$
anchored at the two points of total area at least $\frac{7}{12} \, \area(R)$,
and this bound is the best possible.
\end{lemma}

In order to partition the unit square into strips that contain two
points, one on the boundary, we need to use parity arguments.
Let $P$ be a set of $n$ points in $[0,1]^2$ with $y$-coordinates
$0 \leq y_1 \leq y_2 \leq \cdots \leq y_n \leq 1$.
Set $y_0=0$ and $y_{n+1}=1$.
For $i=1,\ldots,n+1$, put $h_i=y_i -y_{i-1}$, namely $h_i$ is the
$i$th vertical gap. Obviously, we have
\begin{equation} \label{eq:1}
h_i \geq 0 \text{ for all } i=1,\ldots,n+1, \text{ and } \sum_{i=1}^{n+1} h_i =1.
\end{equation}

Parity considerations are handled by the following lemma.
\begin{lemma} \label{lem:partition}
{\rm (i)} If $n$ is odd, then at least one of the following $(n+1)/2$
inequalities is satisfied: 
\begin{equation} \label{eq:2}
h_i + h_{i+1} \leq \frac{2}{n+1}, \text{ for (odd) } i=1,3,\ldots,n-2,n.
\end{equation}

{\rm (ii)} If $n$ is even, then at least one of the following
$n+2$ inequalities is satisfied:
\begin{equation}\label{eq:3}
h_1 \leq \frac{2}{n+2}, \hspace{6mm}
h_{n+1} \leq \frac{2}{n+2},\hspace{6mm}
h_i + h_{i+1} \leq \frac{2}{n+2}, \text{ for } i=1,2,\ldots,n.
\end{equation}
\later{
\begin{align} \label{eq:3}
h_1 &\leq \frac{2}{n+2}, \\ \nonumber
h_{n+1} &\leq \frac{2}{n+2},\\ \nonumber
h_i + h_{i+1} &\leq \frac{2}{n+2}, \text{ for } i=1,2,\ldots,n.
\end{align}
} 
\end{lemma}
\begin{proof}
Assume first that $n$ is odd. Put $a=\frac{2}{n+1}$ and assume that
none of the inequalities in~\eqref{eq:2} is satisfied. Summation
would yield
$ \sum_{i=1}^{n+1} h_i > \frac{n+1}{2} \, a =1, $
which is an obvious contradiction.

Assume now that $n$ is even. Put $a=\frac{2}{n+2}$ and assume that
none of the inequalities in~\eqref{eq:3} is satisfied. Summation would yield
$ 2\sum_{i=1}^{n+1} h_i > (n+2)a =2, $
an obvious contradiction.
\end{proof}

We can now prove the main result of this section.

\begin{theorem} \label{thm:7/12}
For every set of $n$ points in the unit square $[0,1]^2$,
an anchored rectangle packing of area
at least $\frac{7(n-1)}{12(n+1)}$ when $n$ is odd and
$\frac{7n}{12(n+2)}$ when $n$ is even
can be computed in $O(n\log n)$ time.
\end{theorem}
\begin{proof}
Let $P = \{p_1,\ldots ,p_n\}$ be a set of points in the unit square $[0,1]^2$
sorted by their $y$-coordinates with the notation introduced above.
By Lemma~\ref{lem:partition}, we find a horizontal strip
corresponding to one of the inequalities that is satisfied.

Assume first that $n$ is odd. Draw a horizontal line through each
point in $p_j\in P$, for $j$ even, as shown in Fig.~\ref{fig:horizontalstrips}. These
lines divide $[0,1]^2$ into $\frac{n+1}{2}$ rectangles (horizontal
strips).
Suppose now that the satisfied inequality is $h_i + h_{i+1} \leq
\frac{2}{n+2}$ for some odd $i$.
Then we leave a rectangle between $y=y_{i-1}$ and $y=y_{i+1}$ empty,
\ie, $r_i$ is a rectangle of area $0$.
For the remaining rectangles, we assign two consecutive points of $P$
such that each strip above $y=y_{i+1}$ (resp., below $y=y_{i-1}$) is
assigned a point of $P$ on its bottom (resp., top) edge.
Within each rectangle $R$, we can choose two anchored rectangles of
total area at least $\frac{7}{12} \, \area(R)$ by Lemma~\ref{lem:rectangle}.
By Lemma~\ref{lem:partition}(i), the area of the narrowest strip is at
most $\frac{2}{n+1}$. Consequently, the area of the anchored
rectangles is at least $\frac{7}{12}(1-\frac{2}{n+1})
=\frac{7(n-1)}{12(n+1)}$.

Assume now that $n$ is even. If the selected horizontal strip
corresponds to the inequality $h_1 \leq \frac{2}{n+2}$, then divide
the unit square along the lines $y=y_i$, where $i$ is odd. We leave
the strip of height $h_1$ empty, and assign pairs of points to all
remaining strips such that one of the two points lies on the top edge
of the strip. We proceed analogously if the inequality $h_{n+1} \leq
\frac{2}{n+2}$ is satisfied.
Suppose now that the satisfied inequality is $h_i + h_{i+1} \leq \frac{2}{n+2}$.
If $i$ is odd, we leave the strip of height $h_i\leq \frac{2}{n+2}$
(between $y=y_{i-1}$ and $y=y_i$) empty;
if $i$ is even, we leave the strip of height $h_{i+1}\leq
\frac{2}{n+2}$ (between $y=y_i$ and $y_{i+1}$) empty. Above and below
the empty strip, we can form a total of $n/2$ strips, each containing
two points of $P$, with one of the two points lying on the bottom or
the top edge of the strip.
By Lemma~\ref{lem:partition}(i), the area of the empty strip is at most
$\frac{2}{n+2}$. Consequently, the area of the anchored rectangles is
at least $\frac{7}{12}(1-\frac{2}{n+2}) =\frac{7n}{12(n+2)}$, as claimed.
\end{proof}

\section{Lower Bound for Anchored Square Packings} \label{sec:lower:sq}

Given a set $P \subset U=[0,1]^2$ of $n$ points, we show there is an anchored
square packing of large total area. The proof we present is constructive;
we give a recursive partitioning algorithm (as an inductive argument)
based on a quadtree subdivision of $U$ that finds $n$ anchored squares
of total area at least $5/32$. We need the following easy fact:
\begin{observation}\label{obs:adjacent1}
Let $u,v \subseteq U$ be two congruent axis-aligned interior-disjoint squares
sharing a common edge such that
$u \cap P \neq \emptyset$ and $\inter(v)\cap P = \emptyset$.
Then $u \cup v$ contains an anchored empty square whose area is at least $\area(u)/4$.
\end{observation}
\begin{proof}
Let $a$ denote the side-length of $u$ (or $v$).
Assume that $v$ lies right of $u$. Let $p\in P$ be the rightmost point in $u$.
If $p$ lies in the lower half-rectangle of $u$ then the square of side-length $a/2$
whose lower-left anchor is $p$ is empty and has area $a^2/4$.
Similarly, if $p$ lies in the higher half-rectangle of $u$ then the square
of side-length $a/2$ whose upper-left anchor is $p$ is empty and has area $a^2/4$.
\end{proof}

\begin{theorem}\label{thm:5/32}
For every set of $n$ points in $U=[0,1]^2$, where $n \geq 1$,
an anchored square packing of total area at least $5/32$
can be computed in $O(n \log{n})$ time.
\end{theorem}
\begin{proof}
We first derive a lower bound of $1/8$ and then sharpen it to $5/32$.
We proceed by induction on the number of points $n$ contained in
$U$ and assigned to $U$; during the subdivision process, the
r\^ole of $U$ is taken by any subdivision square.
If all points in $P$ lie on $U$'s boundary, $\partial U$, pick one
arbitrarily, say, $(x,0)$; assume $x \leq 1/2$.
(All assumptions in the proof are made without loss of generality.)
Then the square $[x,x +1/2] \times [0,1/2]$ is empty and its area
is $1/4 > 5/32$, as required. Otherwise, discard the points in $P \cap \partial U$
and continue on the remaining points.

If $n=1$, we can assume that $x(p),y(p) \leq 1/2$.
Then the square of side-length $1/2$ whose lower-left anchor is $p$
  is empty and contained in $U$, as desired; hence $A_\textup{sq}(P)\geq 1/4$.
  If $n=2$ let $x_1,x_2, x_3$ be the widths of the
  $3$ vertical strips determined by the two points, where $x_1 + x_2 +x_3=1$.
  We can assume that $0 \leq x_1 \leq x_2 \leq x_3$; then there are
  two anchored empty squares with total area at least
  $x_2^2 + x_3^2 \geq 2/9 > 5/32$, as required.

Assume now that $n \geq 3$. Subdivide $U$ into four congruent squares,
$U_1,\ldots,U_4$, labeled counterclockwise around the center of $U$
according to the quadrant containing the square.
Partition $P$ into four subsets $P_1,\ldots,P_4$ such that $P_i \subset U_i$
for $i=1,\ldots,4$, with ties broken arbitrarily. 
We next derive the lower bound $A_\textup{sq}(P) \geq 1/8$.
We distinguish $4$ cases, depending on the number of empty sets $P_i$.

\smallskip\noindent{\bf Case 1: precisely one of $P_1,\ldots,P_4$ is empty.}
We can assume that $P_1=\emptyset$.
By Observation~\ref{obs:adjacent1}, $U_1\cup U_2$ contains an empty square
anchored at a point in $P_1\cup P_2$ of area at least $\area(U_1)/4=1/16$. By
 induction, $U_3$ and $U_4$ each contain an anchored square packing of area
 at least $c \cdot \area(U_3) = c \cdot \area(U_4)$.
 Overall, we have $A_\textup{sq}(P) \geq 2c/4 + 1/16 \geq c$, which holds
 for $c \geq 1/8$.

\smallskip\noindent{\bf Case 2: precisely two of $P_1,\ldots,P_4$ are empty.}
We can assume that the pairs $\{P_1,P_2\}$ and $\{P_3,P_4\}$
each consist of one empty and one nonempty set. By Observation~\ref{obs:adjacent1}
$U_1\cup U_2$ and $U_3\cup U_4$, respectively, contain a square anchored at a point in
$P_1\cup P_2$ and $P_3\cup P_4$ of area at least $\area(U_0)/4=1/16$.
Hence $A_\textup{sq}(P)\geq 2\cdot \frac{1}{16}=1/8$.

\smallskip\noindent{\bf Case 3: precisely three of $P_1,\ldots,P_4$ are empty.}
We can assume that $P_3 \neq \emptyset$.
Let $(a,b)\in P$ be a maximal point in the product order
(e.g., the sum of coordinates is maximum).
Then $s=[a,a+\frac12]\times [b,b+\frac12]$ is a square anchored at
$(a,b)$, $s\subseteq [0,1]^2$ since $(a,b)\in U_3$,
and $\inter(s)\cap P=\emptyset$. Hence $A_\textup{sq}(P)\geq \area(s)=1/4$.

\smallskip\noindent{\bf Case 4: $P_i \neq \emptyset$ for every $i=1,\ldots,4$.}
Note that $A_\textup{sq}(P)\geq \sum_{i=1}^4 A_\textup{sq}(P_i)$,
where the squares anchored at $P_i$ are restricted to $U_i$.
Induction completes the proof in this case.

\smallskip
In all four cases, we have verified that $A_\textup{sq}(P)\geq 1/8$, as claimed.
The inductive proof can be turned into a recursive algorithm
based on a quadtree subdivision of the $n$ points,
which can be computed in $O(n\log n)$ time~\cite{A05,C83}.
In addition, computing an extreme point (with regard to a specified axis-direction)
in any subsquare over all needed such calls can be executed within
the same time bound. Note that the bound in Case 3 is at least $5/32$
and Case 4 is inductive.

Sharpening the analysis of Cases 1 and 2
yields an improved bound $5/32$; since $5/32 <1/4$, the value $5/32$ is not a
bottleneck for Cases 3 and 4. Details are given in Appendix~\ref{app:5/32};
the running time remains $O(n\log n)$.
\end{proof}

\section{Constant-Factor Approximations for Anchored Square Packings}
\label{sec:approx}

In this section we investigate better approximations for square packings.
Given a finite point set $P\subset [0,1]^2$, perhaps the most natural greedy strategy
for computing an anchored square packing of large area is the following.

\begin{algorithm}\label{alg:sq}{\rm
Set $Q = P$ and $S=\emptyset$. While $Q \neq \emptyset$,
repeat the following. For each point $q \in Q$,
compute a \emph{candidate} square $s(q)$ such that
(i)~$s(q) \subseteq [0,1]^2$ is \emph{anchored} at $q$,
(ii)~$s(q)$ is empty of points from $P$ in its interior,
(iii)~$s(q)$ is interior-disjoint from all squares in $S$, and
(iv)~$s(q)$ has maximum area.
Then choose a largest candidate square $s(q)$, and a corresponding
point $q$, and set $Q \leftarrow Q\setminus \{q\}$
and $S \leftarrow S\cup \{s(q)\}$. When $Q=\emptyset$, return the set
of squares~$S$.
} 
\end{algorithm}

\paragraph{Remark.}
Let $\rho_{\ref{alg:sq}}$ denote the approximation ratio of Algorithm~\ref{alg:sq},
if it exists. The construction in Fig.~\ref{fig:rho2}(a--b) shows that
$\rho_{\ref{alg:sq}} \leq 1/4$. For a small $\eps>0$, consider the point set
$P = \{p_1,\ldots ,p_n\}$, where $p_1 = (1/2+\eps,1/2+\eps)$, $p_2 = (1/2,0)$,
$p_3 = (0,1/2)$, and the rest of the points lie on the lower side of $U$ in the
vicinity of $p_2$, \ie, $x_i \in (1/2-\eps/2, 1/2+\eps/2)$ and $y_i=0$
for $i=4,\ldots,n$.
The packing generated by Algorithm~\ref{alg:sq} consists of a single square of area
$(1/2+\eps)^2$, as shown in Fig.~\ref{fig:rho2}(a), while the packing in
Fig.~\ref{fig:rho2}(b) has an area larger than $1-\eps$.
By letting $\eps$ be arbitrarily small, we deduce that $\rho_{\ref{alg:sq}} \leq 1/4$.
\begin{figure}[htbp]
\centerline{\epsfxsize=6in \epsffile{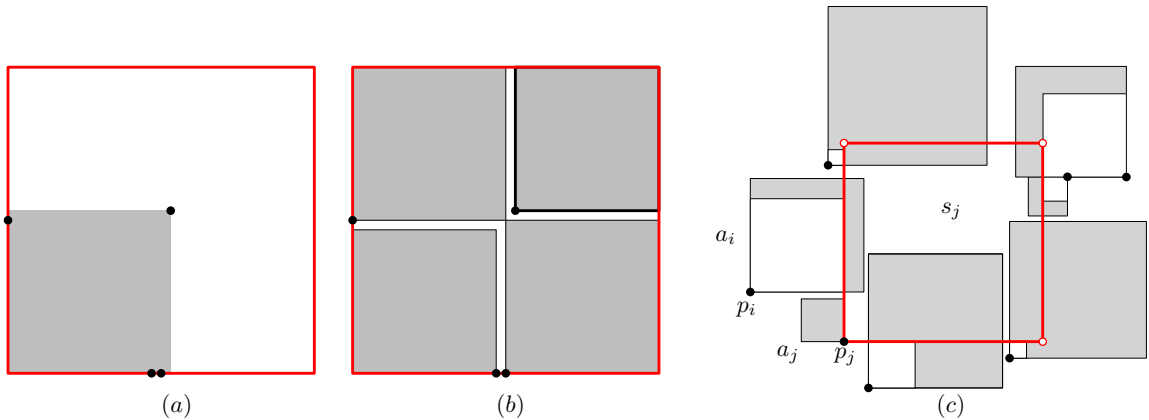}}
\caption{
(a--b) A $1/4$ upper bound for the approximation ratio of Algorithm~\ref{alg:sq}.
(c) Charging scheme for Algorithm~\ref{alg:sq}.
Without loss of generality, the figure illustrates
the case when $s_j$ is a lower-left anchored square.}
\label{fig:alg-sq}\label{fig:rho2}
\end{figure}

We first show that Algorithm~\ref{alg:sq} achieves a ratio of $1/6$ (Theorem~\ref{thm:1/6})
using a charging scheme that compares the greedy packing with an optimal packing.
We then refine our analysis and sharpen the approximation ratio to
$\frac{9}{47} =1/5.22\ldots$ (Theorem~\ref{thm:9/47}).

\paragraph{Charging scheme for the analysis of Algorithm~\ref{alg:sq}.}
Label the points in $P=\{p_1,\ldots , p_n\}$ and the squares in
$S=\{s_1,\ldots , s_n\}$ in the order in which they are processed
by Algorithm~\ref{alg:sq} with $q=p_i$ and $s_i=s(q)$.
Let $G=\sum_{i=1}^n\area(s_i)$ be the area of the greedy packing,
and let $\opt$ denote an optimal packing with
$A =\area(\opt)= \sum_{i=1}^n \area(a_i)$,
where $a_i$ is the square anchored at $p_i$.

We employ a charging scheme, where we distribute the area of every
optimal square $a_i$ with $\area(a_i)>0$ among some greedy squares;
and then show that the total area of the optimal squares charged to
each greedy square $s_j$ is at most $6 \, \area(s_j)$ for all $j=1,\ldots ,n$.
(Degenerate optimal squares, \ie, those with $\area(a_i)=0$ do not need to be charged).
For each step $j=1,\ldots , n$ of Algorithm~\ref{alg:sq}, we shrink some of the squares
$a_1,\ldots , a_n$, and charge the area-decrease to the greedy square $s_j$.
By the end (after the $n$th step), each of the squares $a_1,\ldots ,a_n$
will be reduced to a single point.

Specifically in step $j$, Algorithm~\ref{alg:sq} chooses a square $s_j$, and:
(1) we shrink square $a_j$ to a single point; and
(2) we shrink every square $a_i$, $i>j$ that intersects $s_j$ in
its interior until it no longer does so.
This procedure ensures that no square $a_i$, with $i<j$, intersects $s_j$ in
its interior in step $j$. Refer to Fig.~\ref{fig:alg-sq}(c).
Observe three important properties of the above iterative process:
\begin{itemize} \itemsep 0pt
\item[{\rm (i)}]  After step $j$, the squares $s_1,\ldots,s_j, a_1,\ldots,a_n$
  have pairwise disjoint interiors.
\item[{\rm (ii)}] After step $j$, we have $\area(a_j)=0$ (since $a_j$
  was shrunk to a single point).
\item[{\rm (iii)}] At the beginning of step $j$, if $a_i$ intersects
  $s_j$ in its interior (and so $i \geq j$), then $\area(a_i) \leq \area(s_j)$
  since $s_j$ is feasible for $p_i$ when $a_j$ is selected by Algorithm~\ref{alg:sq}
  due to the greedy choice.
\end{itemize}

\begin{lemma} \label{lem:ratio}
Suppose there exists a constant $\varrho \geq 1$ such that for every
$j=1\ldots n$, square $s_j$ receives a charge of at most $\varrho\, \area(s_j)$.
Then Algorithm~\ref{alg:sq} computes an anchored square packing whose area
$G$ is at least $1/\varrho$ times the optimal.
\end{lemma}
\begin{proof}
Overall, each square $s_j$ receives a charge of at most $\varrho \, \area(s_j)$
from the squares in an optimal solution. Consequently,
$ A =\area(\opt) = \sum_{i=1}^n \area(a_i) \leq \varrho \sum_{j=1}^n \area(s_j) =\varrho \, G, $
and thus $G \geq A/\varrho$, as claimed.
\end{proof}

In the remainder of this section, we bound the charge received by one square $s_j$,
for $j=1,\ldots n$. We distinguish two types of squares $a_i$, $i>j$,
whose area is reduced by $s_j$:
\begin{itemize} \itemsep 0pt
\item[] $\Q_1=\{a_i: i>j$, the area of $a_i$ is reduced by $s_j$, and $\inter(a_i)$
  contains no corner of $s_j\}$,
\item[] $\Q_2=\{a_i:i>j$, the area of $a_i$ is reduced by $s_j$, and
  $\inter(a_i)$ contains a corner of $s_j\}$.
\end{itemize}

It is clear that if the insertion of $s_j$ reduces the area of $a_i$, $i>j$,
then $a_i$ is in either $\Q_1$ or $\Q_2$. Note that the area of $a_j$ is also
reduced to 0, but it is in neither $\Q_1$ nor $\Q_2$.

\begin{lemma}\label{lem:6}
Each square $s_j$ receives a charge of at most $6 \, \area(s_j)$.
\end{lemma}
\begin{proof}
Consider the squares in $\Q_1$. Assume that $a_i$ intersects the
interior of $s_j$, and it is shrunk to $a'_i$. The area-decrease
$a_i \setminus a'_i$ is an L-shaped region, at least half of which lies
inside $s_j$; see Fig.~\ref{fig:alg-sq}. By property (i), the L-shaped regions
are pairwise interior-disjoint; and hence the sum of their areas is at
most $2 \, \area(s_j)$. Consequently, the area-decrease caused by $s_j$
in squares in $\Q_1$ is at most $2 \, \area(s_j)$.

Consider the squares in $\Q_2$. There are at most three squares $a_i$,
$i>j$, that can contain a corner of $s_j$ since the anchor of $s_j$ is
not contained in the interior of any square $a_i$.
Since the area of each square in $\Q_2$ is at most $\area(s_j)$ by property~(iii),
the area decrease is at most $3 \, \area(s_j)$, and so is the charge received
by $s_j$ from squares.

Finally, $\area(a_j)\leq \area(s_j)$ by property~(iii),
and this is also charged to $s_j$. Overall $s_j$ receives a charge of
at most $6 \, \area(s_j)$. 
\end{proof}

The combination of Lemmas~\ref{lem:ratio} and~\ref{lem:6} readily implies the following.
\begin{theorem} \label{thm:1/6}
Algorithm~\ref{alg:sq} computes an anchored square packing whose area
is at least $1/6$ times the optimal.
\end{theorem}

\paragraph{Refined analysis of the charging scheme.}
We next improve the upper bound for the charge received by $s_j$;
we assume for convenience that $s_j=U=[0,1]^2$.
For the analysis, we use only a few properties of the optimal solution.
Specifically, assume that $a_1,\ldots ,a_m$ are interior-disjoint squares
such that each $a_i$:
(a)~intersects the interior of $s_j$;
(b)~has at least a corner in the exterior of $s_j$;
(c)~does not contain $(0,0)$ in its interior; and
(d)~$\area(a_i) \leq \area(s_j)$.

The intersection of any square $a_i$ with $\partial U$
is a polygonal line on the boundary $\partial U$, consisting of one or two segments.
Since the squares $a_i$ form a packing, these intersections are interior-disjoint.

Let $\Delta_1(x)$ denote the maximum area-decrease of a set of
squares $a_i$ in $\Q_1$, whose intersections with $\partial U$ have
total length $x$. Similarly, let $\Delta_2(x)$ denote the maximum
area-decrease of a set of squares $a_i$ in $\Q_2$, whose intersections with
$\partial U$ have total length $x$. By adding suitable squares to $\Q_1$,
we can assume that $4-x$ is the total length of the intersections $a_i
\cap \partial U$ over squares in $\Q_2$ (\ie, the squares in $\Q_1
\cup \Q_2$ cover the entire boundary of $U$).
Consequently, the maximum total area-decrease is given by
\begin{equation} \label{eq:10}
  \Delta(x)= \Delta_1(x) + \Delta_2(4-x), \text{ and }
  \Delta =\sup_{0 \leq x \leq 4} \Delta(x).
\end{equation}

We now derive upper bounds for $\Delta_1(x)$ and $\Delta_2(x)$ independently,
and then combine these bounds to optimize $\Delta(x)$. Since the total perimeter
of $U$ is 4, the domain of $\Delta(x)$ is $0\leq x\leq 4$.

\begin{lemma} \label{lem:Delta1}
The following inequalities hold:
\begin{align}
  \Delta_1(x) &\leq 2, \label{eq:6} \\
  \Delta_1(x) &\leq x, \label{eq:7} \\
  \Delta_1(x) &\leq 1+(x-1)^2, \text{ for } 1 \leq x \leq 2, \label{eq:8} \\
  \Delta_1(x) &\leq 1 + \frac{\lfloor x \rfloor}{4} + \frac{(x -\lfloor x \rfloor)^2}{4},
\text{ for } 0 \leq x \leq 4. \label{eq:9}
  \end{align}
\end{lemma}
\begin{proof}
  Inequality~\eqref{eq:6} was explained in the proof of Theorem~\ref{thm:1/6}.
  Inequalities~\eqref{eq:7} and~\eqref{eq:8} follow from the fact that
  the side-length of each square $a_i$ is at most~$1$ and from the
 fact that the area-decrease is at most the area (of respective squares);
 in addition, we use the inequality $\sum x_i^2 \leq (\sum x_i)^2$, for
 $ x_i \geq 0$ and $\sum x_i \leq 1$, and the inequality
 $x^2 + y^2 \leq 1 + (x+y-1)^2$, for $0 \leq x,y \leq 1$, and $x+y > 1$.

 Write
\begin{equation} \label{eq:13}
  \Delta_1(x) = \Delta_1^{\textup{in}}(x) + \Delta_1^{\textup{out}}(x),
\end{equation}
where $\Delta_1^{\textup{in}}(x)$ and $\Delta_1^{\textup{out}}(x)$
denote the maximum area-decrease contained in $U$ and the complement of $U$,
respectively, of a set of squares in $\Q_1$ whose intersections with $\partial U$ have
total length $x$, where $0 \leq x \leq 4$.
Obviously, $\Delta_1^{\textup{in}}(x) \leq \area(U)=1$. We next show that
$$ \Delta_1^{\textup{out}}(x) \leq  \frac{\lfloor x \rfloor}{4} +
\frac{(x -\lfloor x \rfloor)^2}{4}, $$
and thereby establish inequality~\eqref{eq:9}.

Consider a square $a_i$ of side-length $x_i \leq 1$ in $\Q_1$.
  Let $z_i$ denote the length of the shorter side of the rectangle $a_i \setminus U$.
  The area-decrease outside $U$ equals $x_i z_i - z_i^2$ and so it is
  bounded from above by $x_i^2/4$ (equality is attained when $z_i=x_i/2$).

  Consequently,
  $$  \Delta_1^{\textup{out}}(x) \leq  \sup \sum_{\substack {0 \leq
      x_i \leq 1\\ \sum x_i=x}} \frac{x_i^2}{4} =
  \frac{\lfloor x \rfloor}{4} + \frac{(x -\lfloor x \rfloor)^2}{4}, $$
where the last equality follows from a standard weight-shifting argument,
and equality is attained when $x$ is subdivided into $\lfloor x \rfloor$
unit length intervals and a remaining shorter interval of length $x - \lfloor x \rfloor$.
\end{proof}

Let $k \leq 3$ be the number of squares $a_i$ in $\Q_2$, where $i>j$.
We can assume that exactly $3$ squares $a_i$, with $i > j$, are in $\Q_2$, one for
each corner except the lower-left anchor corner of $U$, that is, $k=3$;
otherwise the proof of Lemma~\ref{lem:6} already yields an
approximation ratio of $1/5$. Clearly, we have $\Delta_2(x) \leq k \leq 3$, for any $x$.

We first bring the squares in $\Q_2$ into \emph{canonical position}:
$x$ monotonically decreases, $\Delta(x)$ does not decrease, and
properties (a--d) listed earlier are maintained.
Specifically, we transform each square $a_i\in \Q_2$ as follows
(refer to Fig.~\ref{fig:f7}): 
\begin{itemize} \itemsep 0pt
\item Move the anchor of $a_i$ to another corner if necessary so that one of its
  coordinates is contained in the interval $(0,1)$;
\item translate $a_i$ horizontally or vertically so that $a_i\cap U$ decreases to
  a skinny rectangle of width $\eps$, for some small $\eps>0$.
\end{itemize}

\begin{figure}[htbp]
\centerline{\epsfxsize=4.7in \epsffile{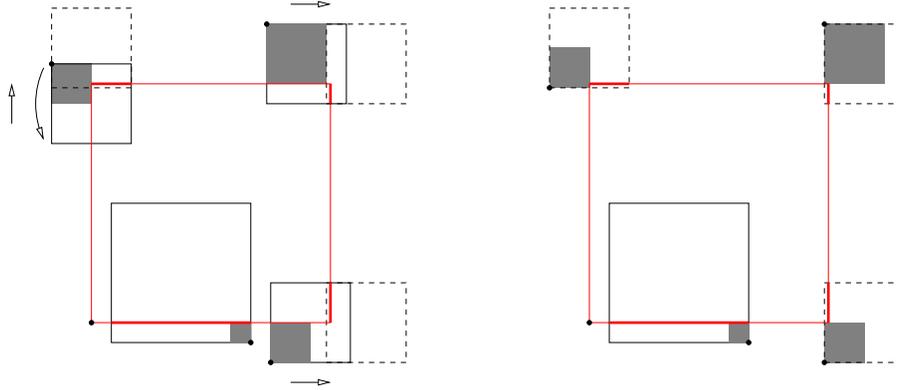}}
\caption{Bounding the area-decrease; moving the squares in $\Q_2$ into canonical position.
The parts of the boundary of $U$ covered by each square (after transformations)
are drawn in thick red lines.}
\label{fig:f7}
\end{figure}

\begin{lemma} \label{lem:Delta2}
The following inequality holds:
\begin{equation}
\Delta_2(x) \leq 2x -\frac{x^2}{3},  \text{ for } 0 \leq x \leq 4. \label{eq:11}
\end{equation}
\end{lemma}
\begin{proof}
Assume that the squares in $\Q_2$ are in canonical position.
Let $y_i$ denote the side-length of $a_i$,
let $x_i$ denote the length of the longer side of the rectangle $a_i \cap U$ and
$z_i$ denote the length of the shorter side of the rectangle $a_i \setminus U$, $i=1,2,3$.
Since the squares in $\Q_2$ are in canonical position, we have $x_i +z_i=y_i \leq 1$,
for $i=1,2,3$. We also have $\sum_{i=1}^3 x_i =x -O(\eps)$. Letting
$\eps \to 0$, we have $\sum_{i=1}^3 x_i =x$.

\begin{align*}
\Delta_2(x) &=\sup_{x_i +z_i=y_i \leq 1} \sum_{i=1}^3 (y_i^2-z_i^2) =
\sup_{0 \leq x_1,x_2,x_3 \leq 1} \sum_{i=1}^3 (1-(1-x_i^2)) =
\sup_{0 \leq x_1,x_2,x_3 \leq 1} \sum_{i=1}^3 (2x_i -x_i^2) \\
&= 2x - \inf_{0 \leq x_1,x_2,x_3 \leq 1} \sum_{i=1}^3 x_i^2 =
2x - 3\frac{x^2}{9} = 2x - \frac{x^2}{3}.
\tag*{\qedhere}
\end{align*}
\end{proof}

Observe that the inequality $ \Delta_2(x) \leq 3$, for every $0 \leq x \leq 4$,
is implied by the inequality~\eqref{eq:11}.
Putting together the upper bounds in Lemmas~\ref{lem:Delta1} and~\ref{lem:Delta2}
yields Lemma~\ref{lem:Delta}:

\begin{lemma} \label{lem:Delta}
The following inequality holds:
\begin{equation}
\Delta \leq \frac{38}{9}. \label{eq:12}
\end{equation}
From the opposite direction, $\Delta \geq 4$ holds even in a geometric
setting, \ie, as implied by several constructions with squares.
\end{lemma}
\begin{proof}
The lower bound $\Delta \geq 4$ is implied by any of the two
(obviously different!) configurations shown in Fig.~\ref{fig:f8}.
\begin{figure}[htbp]
\centerline{\epsfxsize=5in \epsffile{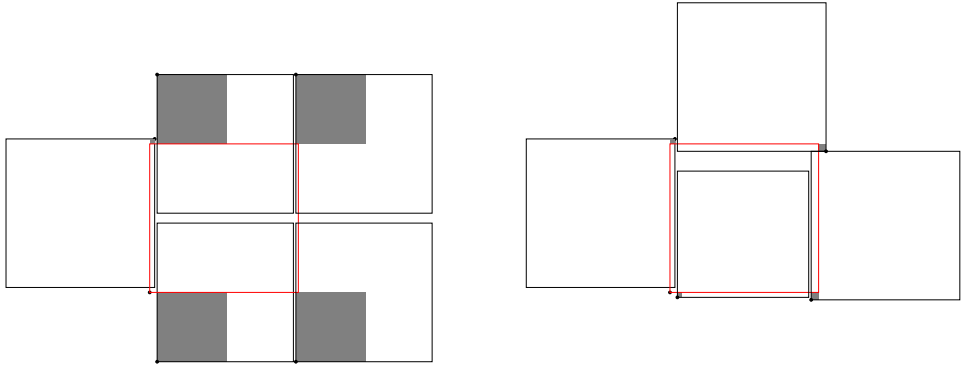}}
\caption{Two constructions with $\Delta \geq 4-\eps$.}
\label{fig:f8}
\end{figure}

We now prove the upper bound.  Recall~\eqref{eq:10}; and by~\eqref{eq:11} we obtain
  $$  \Delta_2(4-x) \leq 2(4-x) - \frac{(4-x)^2}{3}= \frac{8+2x -x^2}{3}. $$
We partition the interval  $[0,4]$ into several
subintervals, and select the best upper bound we have in each case.
We distinguish the following cases (marked with $\clubsuit$):

\medskip
$\clubsuit$ $0 \leq x \leq 1$: By~\eqref{eq:7}, $\Delta_1(x) \leq 1$;
since $\Delta_2(4-x) \leq 3$, we immediately get that
$$ \Delta = \Delta_1(x) + \Delta_2(4-x) \leq 1+3 = 4 < 38/9. $$

\smallskip
$\clubsuit$ $1 \leq x \leq 1 + \frac{\sqrt3}{3}$: Put $g_1(x)=14-4x+2x^2$.
On this interval we have
$$ 1+(x-1)^2 \leq \frac54 + \frac{(x-1)^2}{4}, $$
and so we can use the upper bound in~\eqref{eq:8} on
$\Delta_1(x)$ and obtain
\begin{align*}
\Delta(x) &=\Delta_1(x) + \Delta_2(4-x) \leq 1+(x-1)^2 + \frac{8+2x -x^2}{3}\\
&=\frac{14-4x+2x^2}{3} \leq \frac{g_1(1+1/\sqrt3)}{3}=\frac{38}{9}.
\end{align*}

\smallskip
$\clubsuit$ $1 + \frac{\sqrt3}{3} \leq x < 2$: Put $g_2(x)=50+2x-x^2$.
On this interval we have
$$ 1+(x-1)^2 \geq \frac54 + \frac{(x-1)^2}{4}, $$
and so we can use the upper bound in~\eqref{eq:9} on
$\Delta_1(x)$ and obtain
\begin{align*}
  \Delta(x) &= \Delta_1(x) + \Delta_2(4-x) \leq \frac54 +
  \frac{(x-1)^2}{4} + \frac{8+2x -x^2}{3}\\
&=\frac{50+2x-x^2}{12} \leq \frac{g_2(1+1/\sqrt3)}{12}=\frac{38}{9}.
\end{align*}

\smallskip
$\clubsuit$ $2 \leq x < 3$: Put $g_3(x)=62-4x-x^2$.
Using the upper bound in~\eqref{eq:9} on $\Delta_1(x)$ yields
\begin{align*}
  \Delta(x) &= \Delta_1(x) + \Delta_2(4-x) \leq \frac64 + \frac{(x-2)^2}{4} +
  \frac{8+2x -x^2}{3}\\
&=\frac{62-4x-x^2}{12} \leq \frac{g_3(2)}{12}=\frac{25}{6} < \frac{38}{9}.
\end{align*}

\smallskip
$\clubsuit$ $3 \leq x <4$: Put $g_4(x)=80-10x-x^2$.
Using the upper bound in~\eqref{eq:9} on $\Delta_1(x)$ yields
\begin{align*}
  \Delta(x) &= \Delta_1(x) + \Delta_2(4-x) \leq \frac74 + \frac{(x-3)^2}{4} +
  \frac{8+2x -x^2}{3}\\
&=\frac{80-10x-x^2}{12} \leq \frac{g_4(3)}{12}=\frac{41}{12} < 4 < \frac{38}{9}.
\end{align*}

\smallskip
$\clubsuit$ $x=4$: By~\eqref{eq:6}, $\Delta \leq \Delta_1(x) \leq 2 < 38/9$.

\smallskip\noindent
All cases have been checked, and so the proof of the upper bound is complete.
\end{proof}

\begin{lemma}\label{lem:47/9}
Each square $s_j$ receives a charge of at most $\frac{47}{9}\, \area(s_j)$.
\end{lemma}
\begin{proof}
By Lemma~\ref{lem:Delta}, the area-decrease is at most $38/9 \, \area(s_j)$,
and so is the charge received
by $s_j$ from squares in $\Q_1$ and from squares in $\Q_2$
with the exception of the case $i=j$. Adding this last charge yields a
total charge of at most
$\left( 1 + \frac{38}{9} \right) \, \area(s_j) = \frac{47}{9} \, \area(s_j)$.
\end{proof}

The combination of Lemmas~\ref{lem:ratio} and~\ref{lem:47/9} now
yields the following.
\begin{theorem} \label{thm:9/47}
Algorithm~\ref{alg:sq} computes an anchored square packing whose area
is at least $9/47$ times the optimal.
\end{theorem}

\section{Constant-Factor Approximations for Lower-Left Anchored \\
  Square Packings}
\label{sec:approx:ll}

The following greedy algorithm, analogous to Algorithm~\ref{alg:sq}, constructs
a lower-left anchored square packing, given a finite point set $P\subset [0,1]^2$.

\begin{algorithm}\label{alg:llsq}{\rm
Set $Q = P$ and $S=\emptyset$. While $Q \neq \emptyset$,
repeat the following. For each point $q \in Q$,
compute a \emph{candidate} square $s(q)$ such that
(i)~$s(q) \subseteq [0,1]^2$ has $q$ as its \emph{lower-left} anchor,
(ii)~$s(q)$ is empty of points from $P$ in its interior,
(iii)~$s(q)$ is interior-disjoint from all squares in $S$, and
(iv)~$s(q)$ has maximum area.
Then choose a largest candidate square $s(q)$, and a corresponding
point $q$, and set $Q \leftarrow Q\setminus \{q\}$
and $S \leftarrow S\cup \{s(q)\}$. When $Q=\emptyset$, return the set
of squares~$S$.}
\end{algorithm}

\paragraph{Remark.}
Let $\rho_{\ref{alg:llsq}}$ denote the approximation ratio of Algorithm~\ref{alg:llsq}.
The construction in Fig.~\ref{fig:greedy-lowerleft} shows that $\rho_{\ref{alg:llsq}} \leq 1/3$.
Specifically, for $\eps>0$, with $\eps^{-1}\in \NN$,
consider the point set $P=\{(\eps,\eps), (0,\frac12),(\frac12,0)\}
\cup \{(\frac12+k\eps,\frac12+k\eps) : k=1,\ldots ,1/(2\eps)-1\}$.
Then the area of the packing in Fig.~\ref{fig:greedy-lowerleft}\,(right)
is $\frac34-O(\eps)$, but Algorithm~\ref{alg:llsq} returns the packing shown
in Fig.~\ref{fig:greedy-lowerleft}\,(left) of area $\frac14+O(\eps)$.
\begin{figure}[htbp]
\centerline{\epsfxsize=3in \epsffile{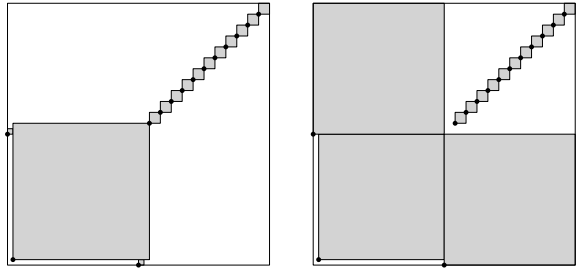}}
\caption{A $1/3$ upper bound for the approximation ratio of Algorithm~\ref{alg:llsq}.}
\label{fig:greedy-lowerleft}
\end{figure}

We next demonstrate that Algorithm~\ref{alg:llsq} achieves approximation ratio $1/3$.
According to the above example, this is the best possible for this algorithm.

\begin{theorem} \label{thm:1/3}
Algorithm~\ref{alg:llsq} computes a lower-left anchored square
packing whose area is at least $1/3$ times the optimal.
\end{theorem}
\begin{proof}
Label the points in $P=\{p_1,\ldots , p_n\}$ and the squares in
$S=\{s_1,\ldots , s_n\}$ in the order in which they are processed by
Algorithm~\ref{alg:llsq} with $q=p_i$ and $s_i=s(q)$.
Let $G=\sum_{i=1}^n\area(s_i)$ be the area of the greedy packing,
and let $\opt$ denote an optimal packing with
$A =\area(\opt)= \sum_{i=1}^n \area(a_i)$, where $a_i$ is the
square anchored at $p_i$.

We charge the area of every optimal square $a_i$ to one or two
greedy squares $s_\ell$; and then show that the total area charged
to $s_\ell$ is at most $3 \, \area(s_\ell)$ for all $\ell=1,\ldots ,n$.
Consider a square $a_i$, $1\leq i\leq n$, with $\area(a_i)>0$.
Let $j=j(i)$ be the minimum index such that $s_j$ intersects the
interior of $a_i$. Let $b_i$ denote the candidate square associated
to $p_i$ in step $j+1$ of Algorithm~\ref{alg:llsq}.
Note that $b_i\subset a_i$, thus $\area(b_i)<\area(a_i)$.
If $\area(b_i)>0$, then let $k=k(i)$ be the minimum index such that
$s_k$ intersects the interior of $b_i$.
\begin{figure}[htbp]
\centerline{\epsfxsize=6in \epsffile{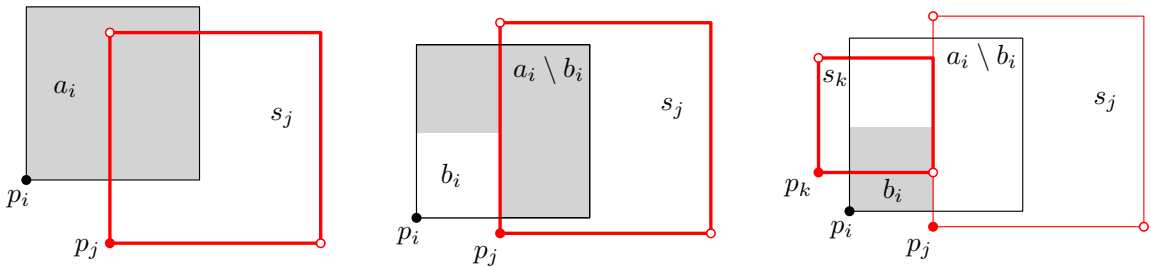}}
\caption{Left: $a_i$ contains the upper-left corner of $s_j$; and
  $\area(a_i)$ is charged to $s_j$.
Middle and Right: $a_i$ contains no corner of $s_j$, but it contains
the lower-right corner of $s_k$.
Then $\area(a_i\setminus b_i)$ is charged to $s_j$ and $\area(b_i)$ is
charged to $s_k$.}
\label{fig:ab}
\end{figure}

We can now describe our \emph{charging scheme}:
If $a_i$ contains the upper-left or lower-right corner of $s_j$,
then charge $\area(a_i)$ to $s_j$ (Fig.~\ref{fig:ab}, left).
Otherwise, charge $\area(a_i\setminus b_i)$ to $s_j$,
and charge $\area(b_i)$ to $s_k$ (Fig.~\ref{fig:ab}, middle-right).

We first argue that the charging scheme is well-defined, and the total
area of $a_i$ is charged to one or two squares ($s_j$ and possibly $s_k$).
Indeed, if no square $s_\ell$, $\ell<i$, intersects the interior of
$a_i$, then $a_i\subseteq s_i$, and $j(i)=i$;
and if $a_i\not\subseteq s_j$ and no square $s_\ell$,
$j<\ell<i$, intersects the interior of
$b_i$, then $b_i\subseteq s_i$ and $k(i)=i$.

\begin{figure}[htbp]
\centerline{\epsfxsize=2.5in \epsffile{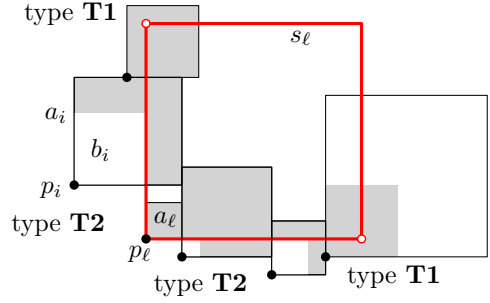}}
\caption{The shaded areas are charged to square $s_\ell$.}
\label{f5+}
\end{figure}

Note that if $\area(a_i)$ is charged to $s_j$, then $\area(a_i) \leq \area(s_j)$.
Indeed, if $\area(a_i) > \area(s_j)$, then $a_i$ is entirely free at step $j$,
so Algorithm~\ref{alg:llsq} would choose a square at least as large as $a_i$ instead of $s_j$,
which is a contradiction. Analogously, if $\area(b_i)$ is charged to $s_k$,
then $\area(b_i)\leq \area(s_k)$. Moreover, if $\area(b_i)$ is charged
to $s_k$, then the upper-left or lower-right corner of~$s_k$ is on the
boundary of $b_i$, and so this corner is contained in $a_i$;
refer to Fig.~\ref{fig:ab}\,(right).

Fix $\ell\in \{1,\ldots , n\}$. We show that the total area charged to
$s_\ell$ is at most $3 \, \area(s_\ell)$.
If a square $a_i$, $i=1,\ldots , n$, sends a positive charge to $s_\ell$,
then $\ell=j(i)$ or $\ell=k(i)$. We distinguish two types of squares $a_i$
that send a positive charge to $s_\ell$; refer to Fig.~\ref{f5+}:
\begin{enumerate}\itemsep -1pt
\item[{\bf T1}\textup{:}] $a_i$ contains the upper-left or lower-right
  corner of $s_\ell$ in its interior.
\item[{\bf T2}\textup{:}] $a_i$ contains neither the upper-left nor
  the lower-right corner of $s_\ell$.
\end{enumerate}

Since $\opt$ is a packing, at most one optimal square contains each corner of $s_\ell$.
Consequently, there is at most two squares $a_i$ of type {\bf T1}.
Since $\area(a_i)\leq \area(s_\ell)$, the charge received from the
squares of type {\bf T1} is at most $2 \, \area(s_\ell)$.

%
The following technical lemma is used for controlling the charge
that a greedy square $s$ receives from squares of an optimal solution.

\begin{lemma}\label{lem:sides}
Let $s$ be an axis-aligned square, and $a_1,a,\ldots , a_m$ be interior-disjoint
axis-aligned squares such that
\begin{itemize}\itemsep 0pt
\item[{\bf P1}\textup{:}]  each $a_i$ intersects the bottom or the left
  side of a square $s$, but
\item[{\bf P2}\textup{:}] the interior of $a_i$ does not contain any corner of $s$.
\end{itemize}
For $i=1,\ldots, m$, let $b_i$ be a maximal square contained in $a_i \setminus \inter(s)$.
Then $\sum_{i=1}^m\area(a_i\setminus b_i)\leq \area(s)$.
\end{lemma}
\begin{proof}
For $i=1, \ldots ,m$, let $c_i$ denote the side-length of $a_i$, and define the
\emph{depth} $d_i$ of $a_i$ as the length of a shortest side of $a_i\cap s$.
With this notation, we have $0\leq d_i\leq c_i\leq c$, the area of
the $L$-shaped region $a_i\setminus b_i$ is
$$ \area(a_i\setminus b_i) = \area(a_i)-\area(b_i)
=c_i^2-(c_i-d_i)^2 =2c_id_i-d_i^2 =(2c_i-d_i)d_i.$$

We apply a sequence of operations on the squares $a_1,\ldots ,a_m$ that
maintain properties {\bf P1} and {\bf P2}, and monotonically increase
$\sum_{i=1}^m (2c_i-d_i)d_i$. The operations successively transform the squares,
eventually eliminate $m-1$ squares, and ensure that the last surviving
square is $a_k=s$.
This implies $\sum_{i=1}^m 2c_id_i-d_i^2\leq \area(s)$, as required.

\begin{figure}[htbp]
\centerline{\epsfxsize=6.3in \epsffile{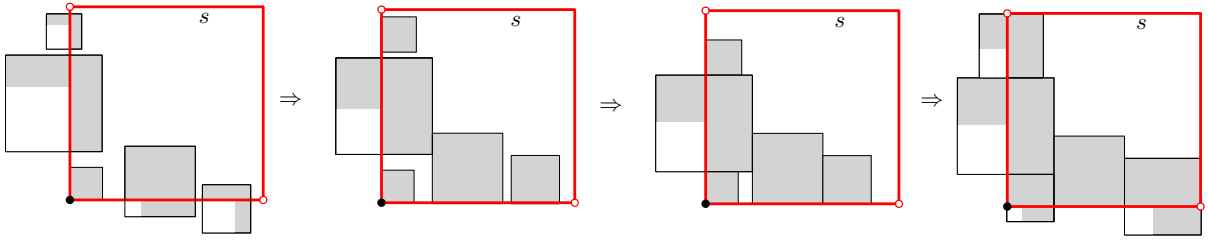}}
\caption{Operations 1, 2 and 3.}
\label{fig:op12}
\end{figure}

\smallskip\noindent{\bf Operation~1.}
Translate $a_i$ horizontally or vertically so that its depth increases
until $c_i=d_i$ or $a_i$ is blocked by some other square.
See Fig.~\ref{fig:op12}. Perform operation~1 successively for every
square $a_i$ in an arbitrary order. The operation can only increase the
contribution $c_i^2-(c_i-d_i)^2$ of $a_i$ since $c_i$ is fixed and the depth
$d_i$ can only increase.

\smallskip\noindent{\bf Operation~2.}
Translate $a_i$ horizontally or vertically
so that its depth remains fixed but $a_i$ moves closer to the
lower-left corner $p$ until it reaches $p$ or is blocked by another square.
See Fig.~\ref{fig:op12}. Perform operation~2 successively for every
square $a_i$ in the order determined by their distance from $p$.
This operation does not change the contribution of the squares $a_i$.

\smallskip\noindent{\bf Operation 3.}
If $a_i$ intersects the bottom (left) side of $s$ and its right (top)
side is not in contact with any other square or a corner of $s$,
then dilate $a_i$ from its upper-left (lower-right)
corner until it is blocked by another square or the boundary of $s$.
See Fig.~\ref{fig:op12}\,(right).

Since operation~3 increases $c_i$ and keeps $d_i$ fixed,
the contribution $(2c_i-d_i)d_i$ of $a_i$ increases.
Note that if the lower-left corner of $a_i$ is $p$ originally
and operation~3 is applied to $a_i$, then its lower-left corner
may move to the exterior of $s$.

\begin{figure}[htbp]
\centerline{\epsfxsize=6.3in \epsffile{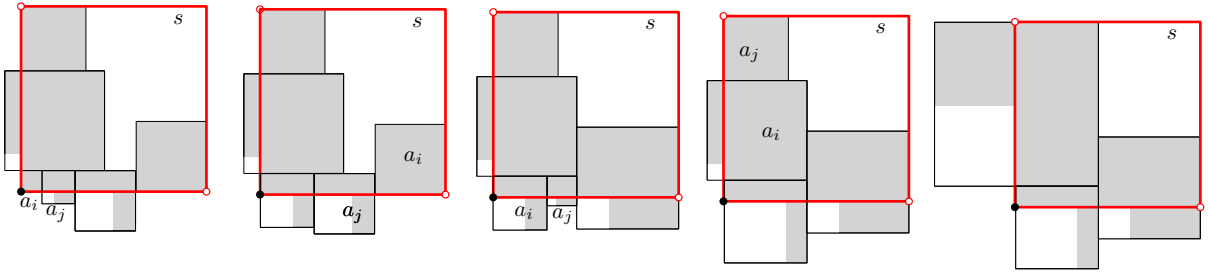}}
\caption{Operation 4.}
\label{fig:op4}
\end{figure}

\smallskip\noindent{\bf Operation~4.}
Consider two squares, $a_i$ and $a_j$,
such that both intersect the bottom side of $s$ and they are adjacent.
Without loss of generality, assume $a_i$ is left of $a_j$ and
$d_i\geq d_j$.  See Fig.~\ref{fig:op4}. We wish to increase $c_i$ and
decrease $c_j$ while $c_i+c_j$ is fixed, and such that their depths
does not increase. We distinguish two cases to ensure that $a_i$
intersects the bottom side of $s$ after the operation.
If $d_j<c_j$, then dilate $a_i$ from its upper-left corner and
$a_j$ from its upper-right corner simultaneously until $d_j=c_j$
or $a_i$ is blocked by some other square or a corner of $s$.
In this case, $d_i$ and $d_j$ remain constant, and the
contribution of $a_i$ and $a_j$ increases or remains the same.
Indeed, if $c_i$ increases and $c_j$ decreases by $\eps>0$, then
\begin{eqnarray*}
(2(c_i+\eps)-d_i)d_i+ (2(c_j-\eps)-d_j)d_j) &=& (2c_i-d_i)d_i+(2c_j-d_j)d_j +2\eps(d_i-d_j)\\
&\geq &  (2c_i-d_i)d_i+(2c_j-d_j)d_j.
\end{eqnarray*}
If $d_j=c_j$, then dilate $a_i$ from its upper-left corner and $a_j$
from its lower-right corner simultaneously until $c_j=0$ or
$a_i$ is blocked by some other square or a corner of $s$.
In this case, $d_i$ is constant, and $d_j=c_j$ decreases.
The contribution of $a_i$ and $a_j$ increases or remains the same:
\begin{eqnarray*}
(2(c_i+\eps)-d_i)d_i+ (d_j-\eps)^2 &=&(2c_i-d_i)d_i+ d_j^2 +2\eps(d_i-d_j)+\eps^2\\
&\geq& (2c_i-d_i)d_i+ d_j^2 +2\eps(d_i-d_j).
\end{eqnarray*}
We apply analogous operations to adjacent squares intersecting the left side of $s$.

\begin{figure}[htbp]
\centerline{\epsfxsize=4.5in \epsffile{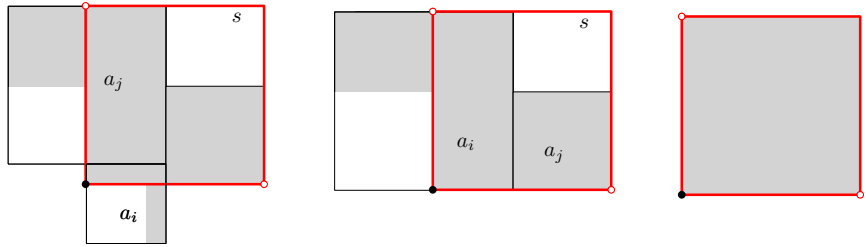}}
\caption{Operation 5.}
\label{fig:op5}
\end{figure}

\smallskip\noindent{\bf Operation~5.}
Consider two adjacent squares, $a_i$ and $a_j$,
such that $a_i$ is incident to the lower-left corner of $s$ and
intersects the bottom (resp., left) side of $s$, and $a_j$ intersects
the left (resp., bottom) side of $s$. See Fig.~\ref{fig:op5}.
Note that $d_j\leq c_i$, otherwise operation~2 could increase $c_i$;
and if $d_j<c_i$, then $c_j=d_j$, otherwise operation~1 could increase $d_j$.
Let $R$ be the axis-aligned bounding box of $(a_i\cup a_j)\cap s$,
which is incident to the lower-left corner
of $p$ and $\area(R)=c_i(d_i+c_j)$. We replace $a_i$ and $a_j$ with one new
square $a_k$ such that $a_k\cap s = R$.

It remains to show that operation~5 increases the total contribution of the squares.
We distinguish two cases depending on which side of $R$ is longer.

\smallskip
\noindent \emph{Case~1.}
If $d_i+c_j\geq c_i$, then $a_k$ has side-length $c_k=d_i+c_j$ and depth $d_k=c_i$.
In case $d_j=c_i$ (Fig.~\ref{fig:op4}, left), we have
\begin{eqnarray*}
(2c_i-d_i)d_i+(2c_j-d_j)d_j
&=& (2c_i-d_i)d_i+(2c_j-c_i)c_i\\
&=& 2c_id_i-d_i^2+2c_jc_i-c_i^2\\
&\leq&  2c_id_i+2c_ic_j-c_i^2\\
&=&(2(d_i+c_j)-c_i)c_i\\
&=&(2c_k-d_k)d_k.
\end{eqnarray*}
Otherwise we have $d_j=c_j<c_i$ and $d_j=c_j$ (Fig.~\ref{fig:op4},
right). Thus,
\begin{eqnarray*}
(2c_i-d_i)d_i+c_j^2
&=& 2c_id_i+c_j^2-d_i^2\\
&\leq& 2c_id_i+c_j^2-(c_i-c_j)^2\\
&=&  2c_id_i+2c_ic_j-c_i^2\\
&=&(2(d_i+c_j)-c_i)c_i\\
&=&(2c_k-d_k)d_k.
\end{eqnarray*}

\smallskip
\noindent \emph{Case~2.}
If $d_i+c_j\leq c_i$, then $a_k$ has side-length $c_k=c_i$ and depth $d_k=d_i+c_j$.
In this case, we have $d_j=c_i$, otherwise operation~1 could increase $d_j$. Thus,
\begin{eqnarray*}
(2c_i-d_i)d_i+c_j^2
&=& 2c_id_i+c_j^2-d_i^2\\
&=& 2c_id_i+2c_j^2-(c_j^2+d_i^2)\\
&=& 2c_id_i+2c_ic_j -2(c_i-c_j)c_j -(d_i^2+c_j^2)\\
&\leq&  2c_id_i+2c_ic_j -2d_ic_j -(d_i^2+c_j^2)\\
&=& 2c_id_i+2c_ic_j-(d_i+c_j)^2\\
&=&(2c_i-(d_i+c_j))(d_i+c_j)\\
&=&(2c_k-d_k)d_k.
\end{eqnarray*}

After successively applying operation~5, we obtain a single square
$a_k=s$, whose contribution is $\area(a_k)=\area(s)$, as
required.
\end{proof}

By Lemma~\ref{lem:sides},
$s_\ell$ receives a charge of at most $\area(s_\ell)$ from squares of type {\bf T2}.
It follows that each $s_\ell$ received a
charge of at most $3 \, \area(s_\ell)$.
Consequently,
$$ A =\area(\opt)
= \sum_{i=1}^n \area(a_i)
\leq 3 \sum_{\ell=1}^n \area(s_\ell)
=3 \, G, $$
and thus $G \geq A/3$.
This completes the proof of Theorem~\ref{thm:1/3}.
\end{proof}

\section{Approximation Schemes}\label{sec:schemes}

In this section, we show that there is a PTAS (resp., QPTAS) for the
maximum anchored square (resp., rectangle) problem
using reductions to the minimum weight independent set (MWIS) problem
for axis-aligned squares (resp., rectangles) in the plane.
The current best approximation ratio for MWIS with axis-aligned rectangles is
$O(\log n/ \log \log n)$~\cite{CH12}, and $O(\log \log n)$ in the unweighted case~\cite{CC09}.
There is a PTAS for MWIS with axis-aligned squares~\cite{EJS01,Cha03} (even a local search
strategy works in the unweighted case~\cite{AM06}); and there is a QPTAS for MWIS with
axis-aligned rectangles~\cite{AW13}. Specifically, for $n$ axis-aligned squares an
$(1-\eps)$-approximation for MWIS which can be computed in time
$n^{O(1/\eps)}$~\cite{Cha03}, and for $n$ axis-aligned rectangles in
time $\exp({\rm poly}(\log n/\eps))$~\cite{AW13}.

\subsection{The Reach of Lower-Left Anchored Packings}

An anchored rectangle can be also viewed from the perspective of robotics applications
such as \emph{a rectangular robotic arm} anchored at the given point.
In this context, an anchored packing of maximum total area represents
the maximum collective \emph{reach} of a collection of nonoverlapping
rectangular robotic arms.
The maximum area of a lower-left anchored rectangle (resp., square) packing can
be arbitrarily small, \ie, with no constant lower-bound guarantee,
when the points in $P$ are close to the top or right boundary of $[0,1]^2$,
or for square packings, when $n$ is large and the points are suitably placed,
\eg, along the diagonal of unit slope.
However, we show below that a greedy packing still covers a constant fraction
of the region that could potentially be covered by lower-left anchored rectangles
(resp., squares). We then employ this lower bound in the design of a QPTAS
(resp., PTAS) for  lower-left anchored rectangles (resp., squares)
in Section~\ref{ssec:approx-schemes}.

\paragraph{The reach of lower-left anchored rectangles.}
For a set $P$ of $n$ points in $U$, define the \emph{reach of lower-left
anchored rectangles} as $W(P)=\bigcup_{p\in P}r(p)$, where $r(p)$ is the axis-aligned
rectangle spanned by $p$ and the upper-right corner of $[0,1]^2$. Clearly,
$W(P)$ is a simple orthogonal polygon in which every lower-left corner is an element in $P$.
Vertical lines through the lower-left corners of $W(P)$ subdivide $W(P)$ into axis-aligned
rectangles whose lower left-corners are in $P$. Each of these rectangles $r$ admits
a lower-left anchored rectangle packing of area at least $0.091\,
\area(r)$ by the main result in~\cite{DT15},
consequently $\area(\opt)\geq 0.091\, \area(W(P))$.

\paragraph{The reach of lower-left anchored squares.}
Similarly, we define the \emph{reach of lower-left
anchored squares} as $W_\textup{sq}(P)=\bigcup_{p\in P}s(p)$, where
$s(p)$ is the (unique) maximal square contained in $U$ whose
lower-left corner is $p$ and that is empty of points of $P$ in its
interior. Observe that $W_\textup{sq}(P)$ is the union of
candidate squares in the first iteration of Algorithm~\ref{alg:llsq}.

\begin{lemma}\label{lem:3s-ll}
Let $P=\{p_1,\ldots,p_n\} \subset U$ be a point set with a lower-left anchored
square packing $S=\{s_1,\ldots ,s_n\}$ returned by Algorithm~\ref{alg:llsq}.
Then $W_\textup{sq}(P)\subset \bigcup_{i=1}^n 3s_i$.
\end{lemma}
\begin{proof}
We need to show that every point $q\in W_\textup{sq}(P)$ lies in $\bigcup_{i=1}^n 3s_i$.
This claim clearly holds when $q$ is covered by one of the squares $s_1,\ldots , s_n$.

Let $q \in W_\textup{sq}(P) \setminus P$ be a point left uncovered by the greedy square packing.
By the definition of $W_\textup{sq}(P)$, point $q$ lies in some
lower-left anchored square $s=s(p_i)$, $p_i\in P$,
of maximum area contained in $[0,1]^2$ but empty of points of $P$ in its interior.
Note that the boundary of $s$ contains at least one point from $P$, namely $p_i\in P$.

Let $s_j$ be the first square chosen by Algorithm~\ref{alg:llsq}
that intersects the interior of $s$.
We are guaranteed the existence of $s_j$ because if no previous greedy
square intersects the interior of $s$,
then $s_i=s$ (possibly $i=j$). At the beginning of the step $j$ of Algorithm~\ref{alg:llsq},
no square in $S$ intersects $s$. By the greedy choice, we have $\area(s_j)\geq \area(s)$.
Since $s_j\cap s\neq \emptyset$ and the side-length of $s$ is less than or equal to the
side-length of $s_j$, we have $q \in s \subset 3s_j$.
\end{proof}

\begin{theorem}\label{thm:1/9-ll}
For every finite point set $P$ in $[0,1]^2$, there is a lower-left anchored square
packing of total area at least $\frac{1}{9}\,\area(W_\textup{sq}(P))$.
\end{theorem}
\begin{proof}
Given a point set $P = \{p_1,\ldots ,p_n\}\subset [0,1]^2$ and a
greedy lower-left anchored square packing $S=\{s_1,\ldots  ,s_n\}$,
Lemma~\ref{lem:3s-ll} yields
$$ \area(W_\textup{sq}(P))\leq \area \left(\bigcup_{i=1}^n 3s_i\right)\leq
\sum_{i=1}^n \area(3s_i)= 9\sum_{i=1}^n \area(s_i), $$
which implies $\sum_i^n\area(s_i) \geq \frac19 \, \area(W_\textup{sq}(P))$, as claimed.
\end{proof}

\paragraph{The reach of anchored squares.}
For a finite (nonempty) set $P\subset [0,1]^2$, the \emph{reach of anchored squares}
$R_\textup{sq}(P)$ is the union of all maximal squares anchored at a
point in $P$ and contained in $U$. Obviously, we have $R_\textup{sq}(P) \geq 5/32$
by Theorem~\ref{thm:5/32}, and we suspect that $\area(R_\textup{sq}(P))\geq \frac12$
for every~$P$. If true, this bound would be the best possible:
choose a small $\eps>0$ and place all $n$ points in the $\eps$-neighborhood
of $(\frac{1}{2},0)$ on the lower side of $U$.
The side-length of any anchored square is at most $\frac{1}{2}+\eps$, which is an upper bound
for the distance of these points from the left and right side of $[0,1]^2$.
Consequently all anchored squares lie below the horizontal line $y=\frac{1}{2}+\eps$,
and so the area of their reach is at most $\frac{1}{2}+\eps$.

\subsection{The approximation schemes}
\label{ssec:approx-schemes}

\paragraph{Anchored rectangles.}
The anchored rectangle packing problem can be discretized~\cite{BT15},
as one can specify a set of $\Theta(n^3)$ anchored rectangles that contains
the optimum packing: consider the Hanan grid~\cite{Ha66} induced by the points,
\ie, the union of $\partial U$ and the $2n$ unit horizontal and vertical unit segments
incident to the $n$ points and contained in $U$. Then an optimal rectangle
packing can be found in this grid (with at most $n^2$ candidates for
each anchor). To formulate the anchored rectangle packing problem
as a maximum weight independent set (MWIS) problem, the set of rectangles
anchored at the same point are slightly enlarged so that they contain
the common anchors in their interiors. For a set of points in general position,
this procedure yields an QPTAS for finding a maximum area
anchored rectangle packing for a set of $n$ points in $U$.

\paragraph{Anchored squares.}
It is worth noting that no similar discretization is known for the anchored
square packing problem. With an $\eps$-factor loss of the total area, however,
we can construct $O(n\log(n/\eps))$ squares
for which MWIS gives an anchored packing of area at least $(1-\eps)\area(\opt)$.
After a similar transformation ensuring that the set of squares anchored at
the same point contain the point in their interiors, and under the assumption
of points in general position, we obtain a PTAS for the maximum-area anchored
square packing problem.

\begin{theorem}\label{thm:ptas-sq}
There is a PTAS for the maximum area anchored square problem.
For every $\eps>0$, there is an algorithm that computes,
for a set of $n$ points in general position in $[0,1]^2$, an anchored square packing
of area at least $(1-\eps)\area(\opt)$ in time $n^{O(1/\eps)}$.
\end{theorem}
\begin{proof}
First we show that the problem can be discretized if we are willing
to lose a factor $\eps>0$ of the maximum area of a packing. Let
$P$ be a set of $n$ points in $[0,1]^2$, and let $\opt$ be an anchored
square packing of maximum area. For a given $\eps>0$, drop all squares
in $\opt$ of area less than $\eps/n$, and shrink all remaining
squares in $\opt$ so that its anchor remains the same, and its
area decreases by a factor of at most $1+\eps$ to $(\eps/n)(1+\eps)^k$
for some integer $k\geq 0$. Denote by $\opt_\eps$ the resulting anchored
square packing. We have dropped small squares of total area at most $\eps$,
and the area of any remaining square decreased by a factor of at most $1+\eps$.
Consequently, $(\area(\opt)-\eps)/(1+\eps)\leq \area(\opt_\eps)$.
For anchored square packings, we have $\area(\opt)\geq 5/32$ by
Theorem~\ref{thm:5/32}, which yields
\begin{equation}\label{eq:ptas}
(1-32\eps/5) \, \area(\opt)\leq \area(\opt_\eps).
\end{equation}

Let $\eps>0$ be given. For every set $P$ of $n$ points in general
position in $[0,1]^2$, we compute a set $\widehat{C}$ of
$O((n/\eps) \log (n/\eps))$  \emph{candidate} squares, and show that a MWIS of
$\widehat{C}$ (where the weight of each square is its area) has area
at least $(1-\eps) \, \area(\opt)$. For each point $p\in P$, consider a
set $C(p)$ of empty squares anchored at $s$ and contained in $U$
whose areas are of the form $(\eps/n)(1 +\eps)^k$,
for $k=0,\ldots ,\lfloor\log(n/\eps)/\log(1+\eps)\rfloor =O(\eps^{-1}\log(n/\eps))$.
Translate each square in $C(p)$, for all $p\in P$, by a sufficiently small $\delta$
so that each square contains its anchor in its interior, but does not cross any new
horizontal or vertical line through the points in $P$; denote by
$\widehat{C}(p)$ the resulting set of squares.
Let $\widehat{C}=\bigcup_{p\in P} \widehat{C}(p)$.
Note that the squares in $\widehat{C}(p)$ pairwise intersect for each
$p\in P$, hence an MWIS contains precisely one square
(possibly of 0 area) for each point in $P$. The number of squares is
$$ |\widehat{C}|=
\sum_{p\in P} |\widehat{C}(p)|=n \cdot O(\eps^{-1}\log (n/\eps))=O((n/\eps) \log (n/\eps)). $$
By \eqref{eq:ptas}, the area of a MWIS
for the candidate squares in $\widehat{C}$ is at least $(1-\eps) \, \area(\opt)$.

Using the PTAS for MWIS with axis-aligned squares due to Chan~\cite{Cha03}:
an $(1-\eps)$-approximation for a set of $m=O((n/\eps) \log (n/\eps))$ weighted
axis-aligned squares can be computed in time
$m^{O(1/\eps)}=n^{O(1/\eps)}$.
\end{proof}

\paragraph{Lower-left anchored packings.}
Even though the area of \emph{lower-left} anchored square packings is
not bounded by any constant, and Theorem~\ref{thm:5/32} no longer holds,
the PTAS in Theorem~\ref{thm:ptas-sq} can be extended for this variant.
Recall that for a set $P\subset [0,1]^2$, the reach of lower-left
anchored squares, denoted $W_\textup{sq}(P)$, is the union of all maximal
squares in $[0,1]^2$ whose lower-left corner is in $P$, and are empty
from points of $P$ in their interior. It is clear that
$W_\textup{sq}(P)$ contains all squares of a lower-left
anchored square packing for $P$, and we have
$\area(\opt)\geq \frac{1}{9}\, \area(W_\textup{sq}(P))$ by Theorem~\ref{thm:1/9-ll}.
By adjusting the resolution of the candidate squares, we obtain a PTAS
for the maximum area lower-left anchored square packing problem.

Similarly, we obtain a QPTAS for the maximum area lower-left anchored rectangle problem.
Consequently, for $n$ anchors, an $(1-\eps)$-approximation for the
lower-left anchored square packing can be computed in time
$n^{O(1/\eps)}$; and for the lower-left anchored rectangle packing
in time $\exp({\rm poly}(\log n/\eps))$.

\section{Conclusion} \label{sec:conclusion}

We conclude with a few open problems:

\begin{enumerate} \itemsep 1pt
\item Is the problem of computing the maximum-area anchored rectangle
    (respectively, square) packing NP-hard?
\item Is there a polynomial-time approximation scheme for the problem
of computing an anchored rectangle packing of maximum area?
\item What lower bound on $A(n)$ can be obtained by extending
  Lemma~\ref{lem:rectangle} concerning rectangles from $2$ to $3$
  points? Is there a short proof of Lemma~\ref{lem:rectangle}?
\item Does Algorithm~\ref{alg:sq} for computing an anchored square packing
      of maximum area achieve a ratio of $1/4$? By Theorem~\ref{thm:9/47}
      and the construction in Fig.~\ref{fig:alg-sq}, the approximation ratio is
      between $9/47=1/5.22\ldots$ and $1/4$. Improvements beyond the
      $1/5$ ratio are particularly exciting.
\item Is the reach of anchored squares always at least $1/2$?
(Does $\area(R_\textup{sq}(P)) \geq \frac12$ hold for every nonempty set~$P$?)
\later{\item Can Algorithm~\ref{alg:sq} for computing an anchored square
  packing be implemented in $O(n \log n)$ time for $n$ points?
  } 
\item Is $A(n) =\frac{2}{3}$? Is $A_\textup{sq}(n)=\frac{7}{27}$?
\item What upper and lower bounds on $A(n)$ and $A_\textup{sq}(n)$ can
    be established in higher dimensions?
\item A natural variant of anchored squares is one where the anchors
    must be the centers of the squares. What approximation can be
    obtained in this case?
\end{enumerate}

\paragraph{Acknowledgment.} The authors are thankful to Ren\'e Sitters for
constructive feedback on the problems discussed in this paper. In particular,
the preliminary approximation ratio $1/6$ in Theorem~\ref{thm:1/6} incorporates
one of his ideas.

\appendix

\section{Proof of Lemma~\ref{lem:rectangle}} \label{sec:rectangle}

We have $P=\{p_1,p_2\}$ with $p_i=(x_i,y_i)$ for $i\in \{1,2\}$, and $y_1=0$.
We can assume that $x_1 \leq x_2$ by applying a reflection with respect to
a vertical line if necessary.
Consider the grid induced by the two points in $[0,1]^2$, consisting of six
rectangles $r_i$, $i=1,\ldots ,6$, as shown in Fig.~\ref{fig:8packings}.
Put $a_i=\area(r_i)$, $i=1,\ldots ,6$. As a shorthand notation, a union of
rectangles, $r_{i_1} \cup \ldots \cup r_{i_k}$ is denoted $r_{i_1\ldots i_k}$,
and $a_{i_1\ldots i_k} = \area(r_{i_1\ldots i_k}) = a_{i_1}+\ldots +a_{i_k}$
denotes its area.

It is not difficult to show that every maximal rectangle packing for $P=\{p_1,p_2\}$
is one of the eight packings in Fig.~\ref{fig:8packings} (although, we
do not use this fact in the remainder of the proof).
Specifically, $p_1$ is a lower-left or lower-right corner of a rectangle.
If $p_1$ is the lower-right corner of a rectangle of height 1, then $p_2$
can be the anchor of any of $r_2$, $r_2$, $r_5$, and $r_6$, producing the
packings $r_{124}$, $r_{134}$, $r_{145}$, and $r_{146}$. If $p_1$ is the
lower-left corner of a rectangle of height 1, then $p_2$ can be the anchor
of $r_3$ or $r_6$, producing $r_{235}$ and $r_{256}$. If $p_1$ is the lower-left
corner of a rectangle of height $y_2$, then $p_2$ can be the anchor of $r_{12}$
or $r_3$, which yields the packings $r_{1256}$ and $r_{356}$.
If $p_1$ is the lower-right corner of a rectangle of height $y_2$,
then $p_2$ must be the anchor of $r_1$, and we obtain $r_{124}$ again.
\begin{figure}[htbp]
\centerline{\epsfxsize=0.7\linewidth \epsffile{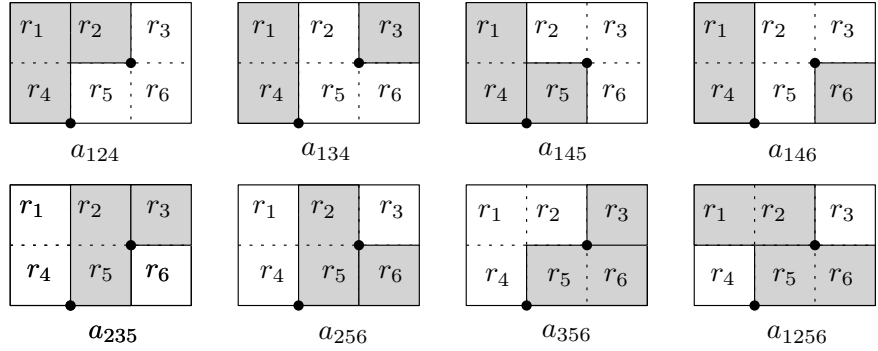}}
\caption{Eight possible anchored rectangle packings for $p_1$ and $p_2$.}
\label{fig:8packings}
\end{figure}

The area of each of the eight packing in Fig.~\ref{fig:8packings} is a
quadratic function in the variables $x_1,x_2,y_2$ over the
domain $$D=\{(x_1,x_2,y_2)\in \mathbb{R}^3: 0\leq x_1\leq x_2\leq 1
\mbox{
  \rm and }0\leq y_2\leq 1\}.$$
We wish to find the minimum $A$ of the upper envelope of these eight
3-variate functions over the domain $D$.
Observe that
\begin{equation} \label{eq:5}
a_{1256} + a_{134} = 1 + a_1.
\end{equation}
In particular, if $a_1 \geq 1/6$ holds, the desired lower bound, $7/12$, follows.
For the case analysis below, we also use the following fact.

\medskip
\emph{Fact 1.
Consider the quadratic equation $12 \tau^2 -7\tau +1=0$. Since its
roots are $\tau_1=1/4$ and $\tau_2=1/3$, we have
$12 \tau^2 -7\tau +1 \leq 0$ for $\tau \in [\tau_1,\tau_2]$. }

\paragraph{Case analysis.}
We distinguish six cases covering the range of $x_1$: $0 \leq x_1 \leq 1/2$.

\medskip
\emph{Case 1:} $x_1 \leq 1/4$. If $y_2 \geq 1/3$ then
$$ \max(a_{356},a_{1256}) \geq (1-x_1) y_2 + \frac{1}{2} (1-y_2) \geq
\frac{3}{4} y_2 + \frac{1}{2} (1-y_2) = \frac{1}{2} + \frac{y_2}{4}
\geq \frac{1}{2} + \frac{1}{12} = \frac{7}{12}, $$
as required.
If $y_2 \leq 1/3$, we distinguish two subcases:

If $x_2 \leq 1/2$, then select rectangle $r_3$ anchored at $p_2$,
and one of $r_{14}$ and $r_{25}$ anchored at $p_1$. If follows that
$$\max(a_{134},a_{235})\geq \frac{x_2}{2} + \frac{2}{3} (1-x_2)
=\frac{2}{3} - \frac{x_2}{6} \geq \frac{2}{3} - \frac{1}{12} = \frac{7}{12}, $$
as required.

If $x_2 \geq 1/2$, then
$$ a_{235} = (1-x_1) - a_6 = (1-x_1) - (1-x_2) y_2
\geq \frac{3}{4} - \frac{1}{2} \cdot \frac{1}{3}=
\frac{7}{12}, $$
as required.

\medskip
\emph{Case 2:} $1/4 \leq x_1 \leq 1/3$.
If $y_2 \geq 1/2$, we have
$$ \max(a_{356},a_{1256}) \geq (1-x_1) y_2 + \frac{1}{2} (1-y_2) \geq
\frac{2}{3} y_2 + \frac{1}{2} (1-y_2) = \frac{1}{2} + \frac{y_2}{6}
\geq \frac{1}{2} + \frac{1}{12} = \frac{7}{12}, $$
as required.

If $y_2 \leq 1/3$, we have
$$ a_1 = x_1 (1-y_2) \geq \frac{1}{4} \cdot \frac{2}{3} = \frac{1}{6}, $$
and the lower bound follows from~\eqref{eq:5}.

If $1/3 \leq y_2 \leq 1/2$, we distinguish two subcases:

\medskip
\emph{Case 2.1:} $x_2 \leq 1/2$.

If $\frac13 \leq y_2 \leq \frac{1}{6(1 -2x_1)}$, then
$$ a_1 = x_1 (1-y_2) \geq x_1 \left( 1 - \frac{1}{6(1 -2x_1)} \right)=
\frac{x_1(5-12 x_1)}{6(1-2x_1)} \geq \frac{1}{6}, $$
and the lower bound follows from~\eqref{eq:5} (due to Fact 1).

If $\frac{1}{6(1 -2x_1)} \leq y_2 \leq \frac12$, then
\begin{align*}
a_{356} &= (1-x_2) + (x_2 -x_1) y_2 \geq
(1-x_2) + (x_2 -x_1) \frac{1}{6(1 -2x_1)} \\
&= \frac{6-13x_1 -5x_2+ 12 x_1 x_2}{6(1-2x_1)} \geq \frac{7}{12},
\end{align*}
as required. The last inequality above is
equivalent to
$$ 5(1-2x_2) \geq 12 x_1 (1-2x_2), $$
which holds for $x_1 \leq 5/12$.

\medskip
\emph{Case 2.2:} $x_2 \geq 1/2$.

If $\frac13 \leq y_2 \leq 1-\frac{1}{6x_1}$, then
$$ a_1 = x_1 (1-y_2) \geq x_1 \, \frac{1}{6x_1} = \frac{1}{6}, $$
and the lower bound follows from~\eqref{eq:5}.

If $1-\frac{1}{6x_1} \leq y_2 \leq \frac12$, then
\begin{align*}
a_{1256} &= (1-x_1) y_2 + x_2 (1-y_2) =
x_2 + y_2 (1-x_1 -x_2) \geq
x_2 + \frac{6x_1 -1}{6x_1} (1-x_1 -x_2) \\
&=\frac{6x_1 x_2 +6x_1 -6 x_1^2 -6x_1 x_2 -1 + x_1 + x_2}{6x_1}
\geq \frac{7x_1 -6 x_1^2 -1/2}{6x_1} \geq \frac{7}{12},
\end{align*}
as required.
Indeed, the last inequality above is equivalent to
$12 x_1^2 -7x_1 +1 \leq 0$, which holds since
$x_1 \in [1/4,1/3]$ and according to Fact 1.

\medskip
\emph{Case 3:} $1/3 \leq x_1 \leq 2/5$.
If $y_2 \leq 1/2$, then
$$ a_1 = x_1 (1-y_2) \geq \frac{1}{3} \cdot \frac{1}{2} = \frac{1}{6}, $$
and the lower bound follows from~\eqref{eq:5}.
If $y_2 \geq 1/2$, we distinguish two subcases:

\medskip
\emph{Case 3.1:} $x_2 \leq 3x_1/2$. Then
\begin{align*}
a_{146} &=x_1 + (1-x_2) y_2 \geq x_1 + \left( 1-\frac{3x_1}{2} \right) y_2 \\
&\geq  x_1 + \frac12 \left( 1-\frac{3x_1}{2} \right) =
\frac{1}{2} + \frac{x_1}{4} \geq \frac{1}{2} + \frac{1}{12}=
\frac{7}{12},
\end{align*}
as required.

\medskip
\emph{Case 3.2:} $x_2 \geq 3x_1/2$. Then
\begin{align*}
a_{1256} &= (1-x_1) y_2 + x_2 (1-y_2) \geq
(1-x_1) y_2 + \frac{3x_1}{2} (1-y_2) \\
&= y_2 + \frac{3x_1}{2} - \frac{5 x_1 y_2}{2} =
\frac{3x_1}{2} + y_2 \left( 1-\frac{5x_1}{2} \right) \\
&\geq  \frac{3x_1}{2} + \frac{1}{2}  \left( 1-\frac{5x_1}{2} \right)
= \frac{1}{2} + \frac{x_1}{4} \geq \frac{1}{2} + \frac{1}{12}=
\frac{7}{12},
\end{align*}
as required.

\medskip
\emph{Case 4:} $2/5 \leq x_1 \leq 5/12$.
Observe that $x_1 \geq 2/5 > 7/19$, and consequently,
$19x_1/7 -1> 0$, as used below in Case 2.4.2. If $y_2 \leq 7/12$, then
$$ a_1 = x_1 (1-y_2) \geq \frac{2}{5} \cdot \frac{5}{12} = \frac{1}{6}, $$
and the lower bound follows from~\eqref{eq:5}.
If $y_2 \geq 7/12$, we distinguish two subcases:

\medskip
\emph{Case 4.1:} $x_2 \leq 12x_1/7$. Then
\begin{align*}
a_{146} &=x_1 + (1-x_2) y_2 \geq x_1 + \left( 1-\frac{12x_1}{7} \right) y_2 \\
&\geq  x_1 + \frac{7}{12} \left( 1-\frac{12x_1}{7} \right) =
\frac{7}{12},
\end{align*}
as required.

\medskip
\emph{Case 4.2:} $x_2 \geq 12x_1/7$. Then
\begin{align*}
a_{1256} &= (1-x_1) y_2 + x_2 (1-y_2) \geq
(1-x_1) y_2 + \frac{12 x_1}{7} (1-y_2) \\
&= \frac{12 x_1}{7} - y_2 \left(\frac{19 x_1}{7} -1 \right)
\geq \frac{12 x_1}{7} - \left(\frac{19 x_1}{7} -1 \right)\\
&= 1-x_1 \geq \frac{7}{12},
\end{align*}
as required.

\medskip
\emph{Case 5:} $5/12 \leq x_1 \leq 4/9$. If $y_2 \leq 3/5$, then
$$ a_1 = x_1 (1-y_2) \geq \frac{5}{12} \cdot \frac{2}{5} = \frac{1}{6}, $$
and the lower bound follows from~\eqref{eq:5}.
If $y_2 \geq 3/5$, we distinguish two subcases:

If $x_2 \leq 13/18$, then
$$ a_{146} =x_1 + (1-x_2) y_2 \geq \frac{5}{12} +
\frac{5}{18} \cdot \frac{3}{5} = \frac{7}{12}, $$
as required.

If $x_2 \geq 13/18$ and $y_2 \leq 5/6$, then
$$ a_{1256} = (1-x_1) y_2 + x_2 (1-y_2) \geq
\frac{5}{9} y_2 + \frac{13}{18} (1-y_2) =
\frac{13}{18} -\frac{y_2}{6} \geq
\frac{13}{18} -\frac{5}{36} = \frac{7}{12}, $$
as required.

If $x_2 \geq 13/18$ and $y_2 \geq 5/6$, then select $r_{45}$ anchored at
$p_2$ and a zero-area rectangle anchored at $p_1$:
$$ a_{45} = x_2 y_2 \geq \frac{13}{18} \cdot \frac{5}{6} =
\frac{65}{108} > \frac{63}{108} = \frac{7}{12}, $$
as required.

\medskip
\emph{Case 6:} $x_1 \geq 4/9$. Select $r_{14}$ and the largest
of the $4$ empty rectangles anchored at $p_2$; the total area is
$$ A \geq x_1 + \frac{1-x_1}{4} = \frac{1+3 x_1}{4} \geq \frac{7}{12}, $$
as required.

This completes the case analysis and thereby the proof of the lower
bound. To see that this lower bound is the best possible,
put $x_1=1/3$ and $x_2=y_2=1/2$ and verify that no system of anchored
rectangles covers an area larger than $7/12$.
\qed

\section{Proof of Theorem~\ref{thm:5/32} (continued)} \label{app:5/32}

The same general idea applies to all cases (as well as the notation in the figures):
two adjacent axis-aligned rectangles are specified, one empty of points and the other
containing points, that allows the selection of an empty square of a suitable size
anchored at an extreme point of the nonemepty rectangle.
The relevant subcases are illustrated in the figures
(nonempty squares or subsquares are shaded) and the compact arguments
are summarized in the tables. Any missing case is either similar or
obtained by reflection of a present case with respect to a line.
In each case:
\begin{itemize} \itemsep 0pt
\item [-] we consider a quadtree subdivision up to a certain level;
\item [-] in nonempty squares in the subdivision (shaded in the figures),
   some anchored square packing covers $c$ times the area of
   corresponding squares by induction;
\item [-] we find new anchored squares of area $1/4$, $1/16$,
  $1/64$, or other, (in bold lines), which partially cover some empty
  and nonempty squares of the subdivision;
\item [-] the new anchored squares supersede the previous packing (\ie, obtained
   by induction) in the nonempty squares of the subdivision that they intersect.
\end{itemize}

We repeatedly make use of the following easy fact:
\begin{observation}\label{obs:adjacent2}
  Let $R_1,R_2 \subseteq U$ be two axis-aligned interior-disjoint rectangles
  of sizes $a_1 \times b$ and $a_2 \times b$,  respectively, where $a_2 \geq b/2$,
  sharing a common edge of length $b$.
  If $R_1 \cap P \neq \emptyset$ and $\inter(R_2)\cap P = \emptyset$,
  then $R_1 \cup R_2$ contains an anchored empty square whose area is at least $b^2/4$.
\end{observation}
\begin{proof}
Assume that $R_1$ lies left of $R_2$ and let $p \in P$ be the rightmost point in $R_1$.
If $p$ lies in the lower half-rectangle of $u$ then the square of side $b/2$
whose lower-left anchor is $p$ is empty and has area $b^2/4$.
Similarly, if $p$ lies in the higher half-rectangle of $R_1$ then the square
of side $b/2$ whose upper-left anchor is $p$ is empty and has area $b^2/4$.
\end{proof}

To specify anchor points, we use the following notation. Let $ \emptyset \neq Q \subset P$.
Let $x^-(Q)$ and $x^+(Q)$ denote the leftmost and rightmost points of $Q$, respectively
(these points coincide if $|Q|=1$).
Let $y^-(Q)$ and $y^+(Q)$ denote the lowest and highest points of $Q$, respectively
(these points coincide if $|Q|=1$).

We next proceed with the details.
In Cases 1 and 2, subdivide $U_i$, $i=1,\ldots,4$ (in Case 1, $U_2$, only) into four
congruent squares $U_{ij}$, $j=1,\ldots,4$, labeled counterclockwise around the center
of $U_i$ according to the subquadrant containing the square.
Then partition $P_i$, $i=1,\ldots,4$ into four subsets $P_{ij}$, $j=1,2,3,4$,
$P_{ij} \subset U_i$ (ties are broken as specified earlier).
We distinguish several subcases depending on the number of empty sets
$P_{ij}$, $i=2$, $j =1,2,3,4$, in the second level of the quadtree subdivision.

\paragraph{Refined analysis of Case 1 ($P_1 =\emptyset$, only).}
Let $\lambda_2 \in \{0,1,2,3\}$ denote the number of empty $P_{2j}$.
The relevant subcases are illustrated in  Fig.~\ref{f10} and summarized
in Table~\ref{tab:case1}. For example, in the first two rows in the table
there are two inductive instances of area $1/4$ and two inductive instances of area $1/16$,
contributing $2c/4$ and $2c/16$, respectively, toward the total covered area.
\begin{figure}[htbp]
\centerline{\epsfxsize=4.5in \epsffile{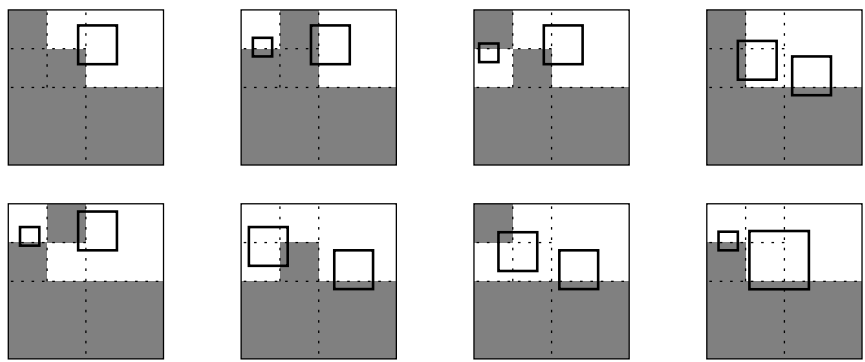}}
\caption{Case 1: the number of empty $P_{2j}$ is in $\{0,1,2,3\}$;
rows: $1,2$; left to right: $a,b,c,d$.}
\label{f10}
\end{figure}
\begin {table} [htbp]
\centering
\begin{tabular}{|c|c|c|c|c|}
\hline
Figure & $\lambda_2$ & Anchor(s) & Inductive/direct proof & Resulting $c$ \\
\hline
\hline
                       & $0$ & $x^+(P_{21} \cup P_{24})$ &
$c= \frac{2c}{4} + \frac{2c}{16} + \frac{1}{16}$ & $\frac{1}{6} > \frac{5}{32} $ \\ \hline
\ref{f10}\,(1a) & $1$ & $x^+(P_{21} \cup P_{24})$ &
$c= \frac{2c}{4} + \frac{2c}{16} + \frac{1}{16}$ & $\frac{1}{6} > \frac{5}{32} $ \\ \hline
\ref{f10}\,(1b,1c) & $2$ & $y^+(P_{23})$, $x^+(P_{21} \cup P_{24})$ &
$c= \frac{2c}{4} + \frac{1}{16} + \frac{1}{64}$ & $\frac{5}{32}$ \\ \hline
\ref{f10}\,(1d) & $2$ & $x^+(P_{22} \cup P_{23})$, $y^+(P_4)$ &
$c= \frac{c}{4} + \frac{2}{16}$ & $\frac{1}{6} > \frac{5}{32} $ \\ \hline
\ref{f10}\,(2a) & $2$ & $y^+(P_{23})$, $x^+(P_{21} \cup P_{24})$ &
$c= \frac{2c}{4} + \frac{1}{16} + \frac{1}{64}$ & $\frac{5}{32}$ \\ \hline
\ref{f10}\,(2b,2c) & $3$ & $x^-(P_{21} \cup P_{24})$, $y^+(P_4)$ &
$c= \frac{c}{4} + \frac{2}{16}$ & $\frac{1}{6} > \frac{5}{32} $ \\ \hline
\ref{f10}\,(2d) & $3$ & $y^+(P_{23})$, $y^+(P_{31} \cup P_{34} \cup P_4)$ &
$c= \frac{9}{64} + \frac{1}{64}$ & $\frac{5}{32}$ \\ \hline
\hline
\end{tabular}
\caption {Inductive/direct proofs for the subcases of Case 1.}
\label {tab:case1}
\end {table}

\vspace{-2\baselineskip}

\paragraph{Refined analysis of Case 2.}
If the two nonempty squares are adjacent, say $P_2, P_3 \neq \emptyset$,
then one of the two squares of side $1/2$, whose lower-left
or upper-left anchor is the rightmost point in $P_2 \cup P_3$ is empty and
contained in $U$; its area is $1/4 > 5/32$, as required.
Assume henceforth that the two nonempty squares lie in opposite quadrants,
\eg, $P_2,P_4 \neq \emptyset$. For $i=2,4$, let $\lambda_i$ denote the number
of empty sets $P_{ij}$; by symmetry, we can assume that $\lambda_2 \leq \lambda_4$;
thus $0 \leq \lambda_2 \leq \lambda_4 \leq 3$. We further distinguish
Cases 2a, 2b, 2c, 2d, as determined by the different values of $\lambda_2$.

\smallskip\noindent{\bf Case 2a: $\lambda_2=0$.}
The relevant subcases are illustrated in  Fig.~\ref{f11} and summarized
in Table~\ref{tab:case2a}. Fig.~\ref{f11}\,(3b) and Fig.~\ref{f11}\,(3c)
correspond to $|P_{42}| \geq 2$ and $|P_{42}|=1$, respectively.
\begin{figure}[htbp]
\centerline{\epsfxsize=4.5in \epsffile{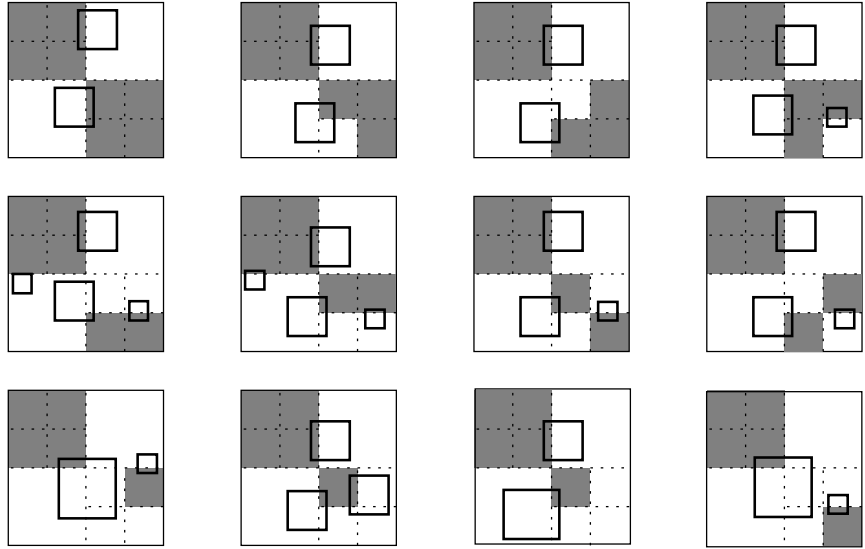}}
\caption{Case 2a: the number of empty $P_{2j}$ is $0$; rows: $1,2,3$;
  left to right: $a,b,c,d$.} 
\label{f11}
\end{figure}
\begin {table} [htbp]
\centering
\begin{tabular}{|c|c|c|c|c|c|}
\hline
Figure & $\lambda_2$ & $\lambda_4$ & Anchor(s) & Inductive/direct proof & Resulting $c$ \\
\hline
\hline
\ref{f11}\,(1a) & $0$ & $0$ & $x^+(P_{21} \cup P_{24})$, $x^-(P_{42} \cup P_{43})$ &
$c= \frac{4c}{16} + \frac{2}{16}$ & $\frac{1}{6} > \frac{5}{32} $ \\ \hline
\ref{f11}\,(1b,1c) & $0$ & $1$ & $x^+(P_{21} \cup P_{24})$, $x^-(P_{42} \cup P_{43})$ &
$c= \frac{4c}{16} + \frac{2}{16}$ & $\frac{1}{6} > \frac{5}{32} $ \\ \hline
\ref{f11}\,(1d) & $0$ & $1$ & $x^+(P_{21} \cup P_{24}), x^-(P_{42} \cup P_{43}), \ldots$ &
$c= \frac{2c}{16} + \frac{2}{16} + \frac{1}{64}$ & $\frac{9}{56} > \frac{5}{32}$ \\ \hline
\ref{f11}\,(2a,2b) & $0$ & $2$ & $x^+(P_{21} \cup P_{24}), y^-(P_{23}), \ldots$ &
$c= \frac{2}{16} + \frac{2}{64}$ & $\frac{5}{32}$ \\ \hline
\ref{f11}\,(2c,2d) & $0$ & $2$ & $x^+(P_{21} \cup P_{24}), x^-(P_{42} \cup P_{43}), \ldots$ &
$c= \frac{2c}{16} + \frac{2}{16} + \frac{1}{64}$ & $\frac{9}{56} > \frac{5}{32}$ \\ \hline
\ref{f11}\,(3a) & $0$ & $3$ & $y^-(P_{23} \cup P_{24}), y^+(P_{41})$ &
$c= \frac{9}{64} + \frac{1}{64}$ & $\frac{5}{32}$ \\ \hline
\ref{f11}\,(3b) & $0$ & $3$ & $x^+(P_{21} \cup P_{24}), x^-(P_{42} \cup P_{43}), \ldots$ &
$c= \frac{3}{16}$ & $\frac{3}{16} > \frac{5}{32}$ \\ \hline
\ref{f11}\,(3c) & $0$ & $3$ & $x^+(P_{21} \cup P_{24}), x^-(P_{42} \cup P_{43})$ &
$c \geq \frac{1}{16} + \frac{25}{256} $ &
$\frac{41}{256} > \frac{5}{32}$ \\ \hline
\ref{f11}\,(3d) & $0$ & $3$ & $y^-(P_{23} \cup P_{24}), y^+(P_{44})$ &
$c \geq \frac{9}{64} + \frac{1}{64}$ & $\geq \frac{5}{32}$ \\ \hline
\hline
\end{tabular}
\caption {Inductive/direct proofs for Case 2a.}
\label {tab:case2a}
\end {table}

\smallskip\noindent{\bf Case 2b: $\lambda_2=1$.}
The relevant subcases for $\lambda_4=1,2,3$ are are illustrated in
Figures~\ref{f12},~\ref{f13},~\ref{f14}, respectively; the proofs are
summarized in Table~\ref{tab:case2b}.

For example: if $\lambda_4=3$, so that $P_{42} \neq \emptyset$
as in Fig.~\ref{f14}\,(1b,2b,3b),
we proceed depending on whether $|P_{42}| \geq 2$ or $|P_{42}|=1$.
If $|P_{42}| \geq 2$ we proceed as in Fig.~\ref{f14}\,(1b);
if $|P_{42}|=1$ and the corresponding point, say $p$, lies
in $U_{422} \cup U_{423}$, then $p$ is set as a left anchor
and one can get an empty covered area at least $9/64 + 1/64 =5/32$;
if $|P_{42}|=1$ and the corresponding point, say $p$, lies
in $U_{421} \cup U_{424}$, then $p$ is set as an upper-left anchor
and the lowest point in $P_{24}$ is set as an upper anchor
as shown in Fig.~\ref{f14}\,(2b,3b), with an empty covered area
at least $25/256 + 1/16 = 41/256 > 5/32$.

\begin{figure}[htbp]
\centerline{\epsfxsize=4.5in \epsffile{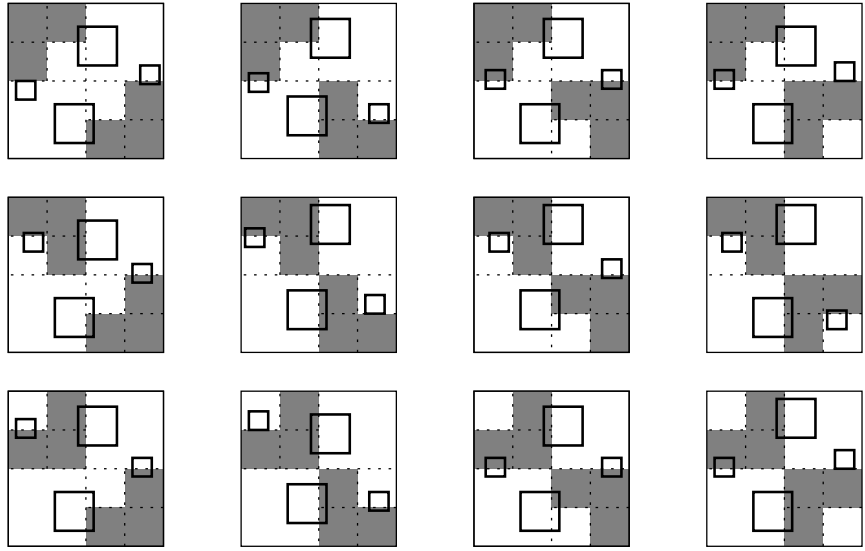}}
\caption{Case 2b with one empty $P_{2j}$ and one empty $P_{4j}$;
rows: $1,2,3$; left to right: $a,b,c,d$.}
\label{f12}
\end{figure}

\begin{figure}[htbp]
\centerline{\epsfxsize=4.5in \epsffile{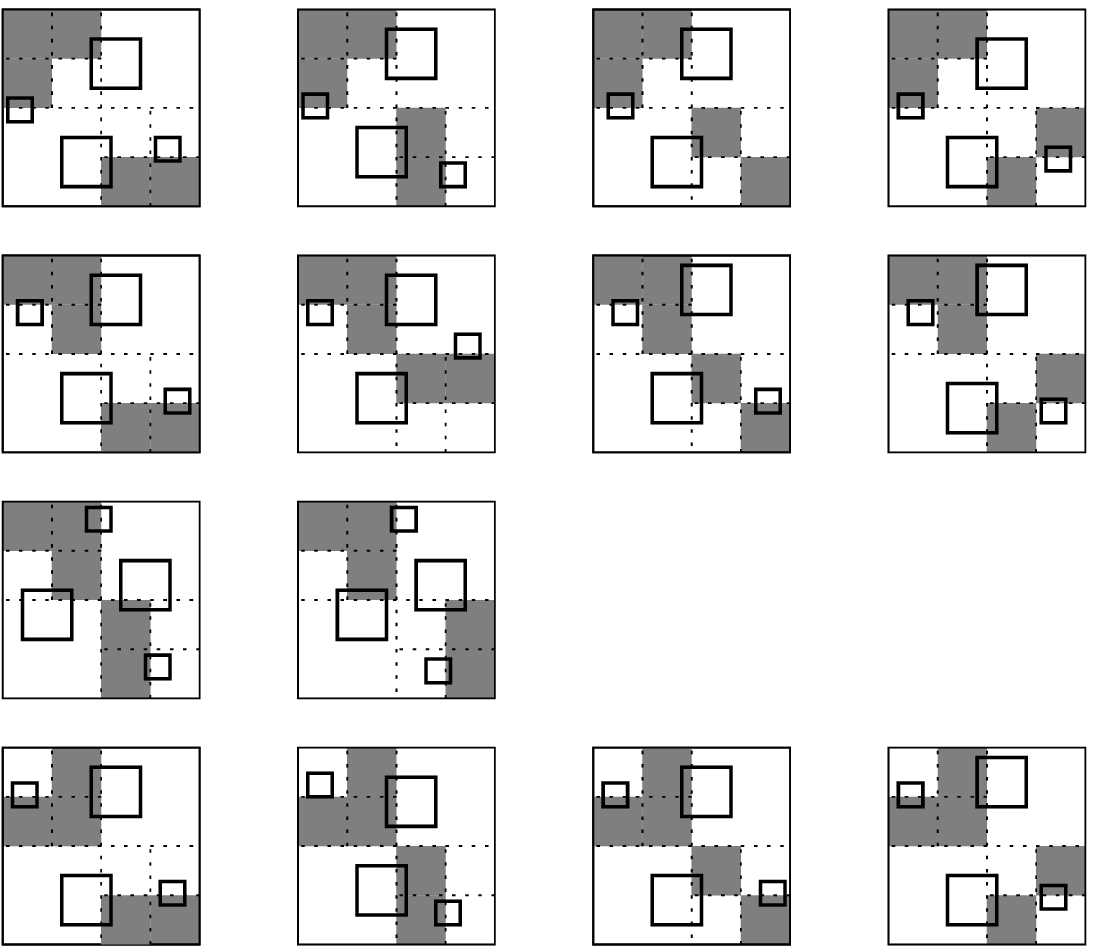}}
\caption{Case 2b with one empty $P_{2j}$ and two empty $P_{4j}$;
rows: $1,2,3,4$; left to right: $a,b,c,d$.}
\label{f13}
\end{figure}

\begin{figure}[htbp]
\centerline{\epsfxsize=4.5in \epsffile{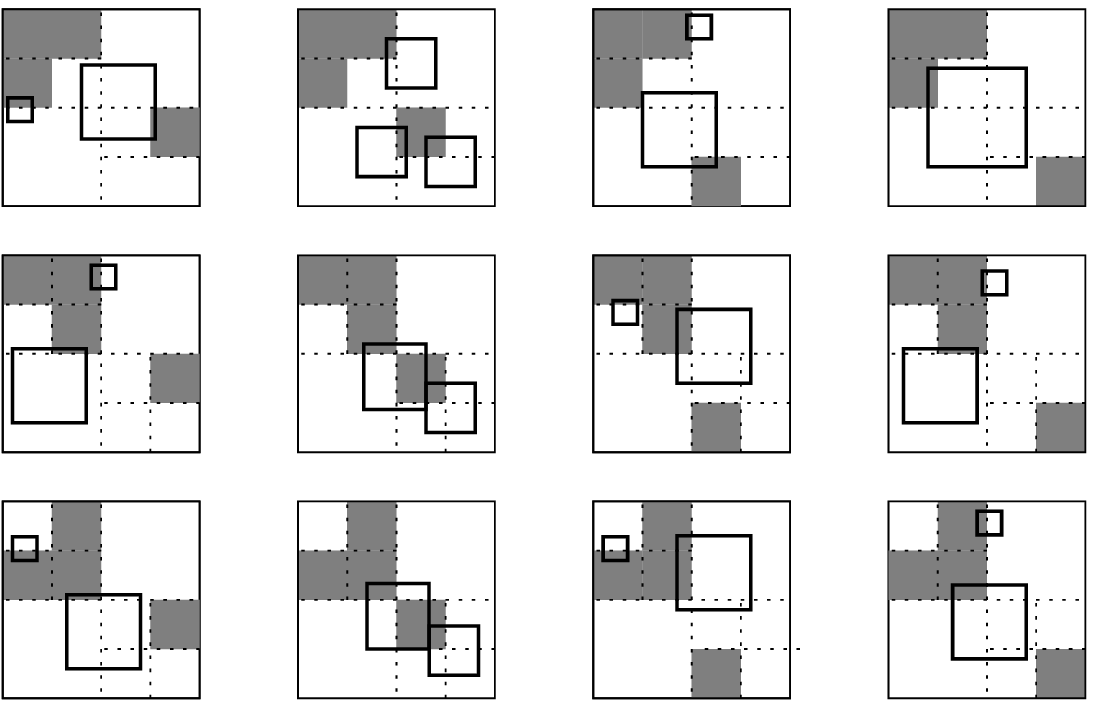}}
\caption{Case 2b with one empty $P_{2j}$ and three empty $P_{4j}$;
rows: $1,2,3$; left to right: $a,b,c,d$.}
\label{f14}
\end{figure}

\begin {table} [htbp]
\centering
\begin{tabular}{|c|c|c|c|c|c|}
\hline
Figure & $\lambda_2$ & $\lambda_4$ &  Anchor(s) & Direct proof & Resulting $c$ \\
\hline
\hline
\ref{f12}\,(1a,1b,1c,1d) & $1$ & $1$ & $x^+(P_{21} \cup P_{24}), y^-(P_{23}), \ldots$ &
$c\geq \frac{2}{16} + \frac{2}{64}$ & $ \geq \frac{5}{32}$ \\ \hline
\ref{f12}\,(2a,2b,2c,2d) & $1$ & $1$ & $x^+(P_{21} \cup P_{24}), y^-(P_{22}), \ldots$ &
$c\geq \frac{2}{16} + \frac{2}{64}$ & $ \geq \frac{5}{32}$ \\ \hline
\ref{f12}\,(3a,3b,3c,3d) & $1$ & $1$ & $x^+(P_{21} \cup P_{24}), y^+(P_{23}), \ldots$ &
$c\geq \frac{2}{16} + \frac{2}{64}$ & $ \geq \frac{5}{32}$ \\ \hline
\ref{f13}\,(1a through 4d)& $1$ & $2$ & &
$c\geq \frac{2}{16} + \frac{2}{64}$ & $ \geq \frac{5}{32}$ \\ \hline
\ref{f14}\,(1a,1c) & $1$ & $3$ & &
$c \geq \frac{9}{64} + \frac{1}{64}$ & $\geq \frac{5}{32}$ \\ \hline
\ref{f14}\,(1b) & $1$ & $3$ & &
$c \geq \frac{25}{256} + \frac{1}{16}$ & $\geq \frac{5}{32}$ \\ \hline
\ref{f14}\,(1d) & $1$ & $3$ & &
$c \geq \frac{1}{4}$ & $\geq \frac{1}{4} > \frac{5}{32} $ \\ \hline
\ref{f14}\,(2a,2c,2d) & $1$ & $3$ & &
$c \geq \frac{9}{64} + \frac{1}{64}$ & $\geq \frac{5}{32}$ \\ \hline
\ref{f14}\,(2b) & $1$ & $3$ & &
$c \geq \frac{25}{256} + \frac{1}{16}$ & $\geq \frac{5}{32}$ \\ \hline
\ref{f14}\,(3a,3c,3d) & $1$ & $3$ & &
$c \geq \frac{9}{64} + \frac{1}{64}$ & $\geq \frac{5}{32}$ \\ \hline
\ref{f14}\,(3b) & $1$ & $3$ & &
$c \geq \frac{25}{256} + \frac{1}{16}$ & $\geq \frac{5}{32}$ \\ \hline\hline
\end{tabular}
\caption {Direct proofs for Case 2b.}
\label {tab:case2b}
\end {table}

\smallskip\noindent{\bf Case 2c: $\lambda_2=2$.}
The relevant subcases are illustrated in  Fig.~\ref{f15} and summarized
in Table~\ref{tab:case2c2d}.
Assume first that $\lambda_4=2$.
If $P_{22}$ and $P_{24}$ are nonempty, as in Fig.~\ref{f15}\,(1a,1b),
we proceed according to whether the two nonempty $P_{4j}$
are in adjacent vertical strips of width $1/4$ or in the same vertical strip
of width $1/4$: in the former case, we select empty squares as in Fig.~\ref{f15}\,(1a)
while in the latter case, we select empty squares as in Fig.~\ref{f15}\,(1b).
If $P_{21}$ and $P_{23}$ are nonempty, as in Fig.~\ref{f15}\,(1c,1d),
we proceed analogously.
If $P_{22}$ and $P_{23}$ are nonempty, as in Fig.~\ref{f15}\,(2a,2b),
it suffices to consider the cases when the two nonempty squares $P_{4j}$
are in the same strip of width $1/4$ (the complementary cases have been
covered in row 1); then the empty covered area is at least $9/64 + 1/64 =5/32$, or
at least $3/16 > 5/32$, as required.
If $P_{21}$ and $P_{24}$ are nonempty, as in Fig.~\ref{f15}\,(2c,2d),
it again suffices to consider the cases when the two nonempty squares $P_{4j}$
are in the same strip of width $1/4$, and then the empty covered area is at least
$3/16 > 5/32$, as required.

\begin{figure}[htbp]
\centerline{\epsfxsize=4.5in \epsffile{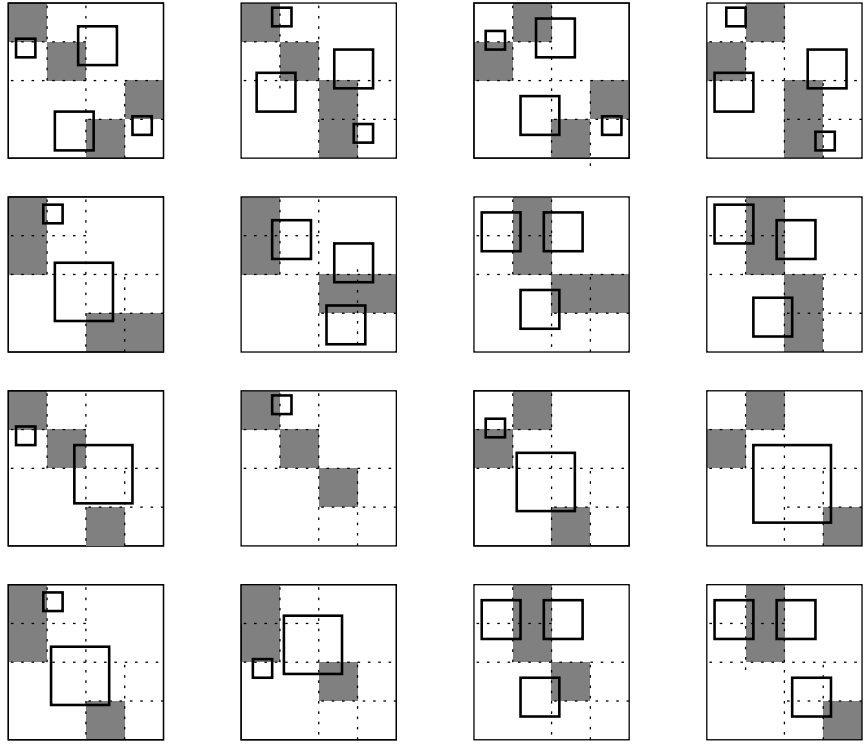}}
\caption{Case 2c: two empty $P_{2j}$ and two or three empty $P_{4j}$;
rows: $1,2,3,4$; left to right: $a,b,c,d$.}
\label{f15}
\end{figure}

Assume now that $\lambda_4=3$.
If $P_{22}$ and $P_{24}$ are nonempty, as in Fig.~\ref{f15}\,(3a,3b),
we proceed as follows\footnote{Observe that our quadtree decomposition analysis
``matches'' the lower bound construction in Proposition~\ref{prop:7/27} in the sense that
  the hardest case to deal with in our proof is the one in which the
  lower bound construction fits in.}: if $P_{41}$, $P_{43}$ or $P_{44}$ are nonempty,
the empty covered area is at least $9/64 + 1/64 =5/32$;
if $P_{42}$ is nonempty, consider $|P_{22}|$, $|P_{24}|$, and $|P_{42}|$.
If $|P_{24}| \geq 2$ or $|P_{42}| \geq 2$,
select two empty squares of area $1/16$ anchored at points of the
respective ``large'' set, and another empty square of area $1/16$ anchored
at the point of the other set (among $P_{24}$ and $P_{42}$);
we thereby have a total empty area $3/16 > 5/32$, as required.
If $|P_{24}|=|P_{42}|=1$, the empty covered area is at least $9/64 + 1/64 =5/32$,
unless the point in $P_{24}$, say $p$, lies in $U_{242}$ and
the point in $P_{42}$, say $q$, lies in $U_{424}$;
then an empty square whose upper-left anchor is $p$ and side $x(q)-x(p)$
and an empty square whose lower-left anchor is $q$ and side $1-x(q)$
cover an area at least
$2 \left(\frac{5}{16} \right)^2 = \frac{25}{128} > \frac{5}{32}$, as required.
If $P_{21}$ and $P_{23}$ are nonempty, as in Fig.~\ref{f15}\,(3c,3d),
the empty covered area is at least $9/64 + 1/64 =5/32$, or at least $1/4$, as required.
If $P_{22}$ and $P_{23}$ are nonempty, as in Fig.~\ref{f15}\,(4a,4b),
there are two possibilities: $P_{4j}$ is nonempty where $j \in \{1,2\}$,
or where $j \in \{3,4\}$; each case is dealt with easily.
If $P_{21}$ and $P_{24}$ are nonempty, as in Fig.~\ref{f15}\,(4c,4d),
each case is dealt with easily.

\smallskip\noindent{\bf Case 2d: $\lambda_2=3$}.
Since $\lambda_2 \leq \lambda_4$, we have $\lambda_4=3$.
The relevant subcases are illustrated in  Fig.~\ref{f16} and summarized
in Table~\ref{tab:case2c2d}.
If $P_{21}$ or $P_{22}$ or $P_{23}$ are nonempty, as in Fig.~\ref{f16}\,(1a,1b),
the empty covered area is at least $9/64 + 1/64 = 5/32$, as required.
If $P_{24}$ is nonempty, as in Fig.~\ref{f16}\,(1c), at least one
of the nonempty sets has at least $2$ points, and so the empty covered area
is at least $3/16 > 5/32$, as required.

\begin{figure}[htbp]
\centerline{\epsfxsize=3.3in \epsffile{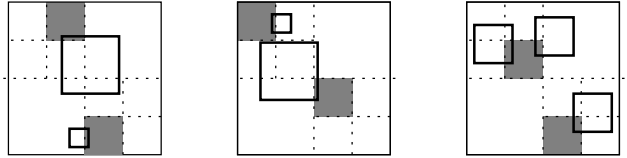}}
\caption{Case 2d: three empty $P_{2j}$ and three empty $P_{4j}$; left to right: $a,b,c$.}
\label{f16}
\end{figure}

\begin {table} [htbp]
\centering
\begin{tabular}{|c|c|c|c|c|}
\hline
Figure & $\lambda_2$ & $\lambda_4$ & Direct proof & Resulting $c$ \\
\hline
\hline
\ref{f15}\,(1a,1b) & $2$ & $2$ & $c \geq \frac{2}{16} + \frac{2}{64}$ &
$ \geq \frac{5}{32}$ \\ \hline
\ref{f15}\,(1c,1d) & $2$ & $2$ & $c \geq \frac{2}{16} + \frac{2}{64}$ &
$ \geq \frac{5}{32}$ \\ \hline
\ref{f15}\,(2a,2b) & $2$ & $2$ & $c \geq \frac{9}{64} + \frac{1}{64}$ &
$ \geq \frac{5}{32}$ \\ \hline
\ref{f15}\,(2c,2d) & $2$ & $2$ & $c \geq \frac{3}{16}$ &
$ \geq \frac{3}{16} > \frac{5}{32}$ \\ \hline
\ref{f15}\,(3c,4a,4b) & $2$ & $3$ & $c \geq \frac{9}{64} + \frac{1}{64} = \frac{5}{32}$ &
$ \geq \frac{5}{32}$ \\ \hline
\ref{f15}\,(3d) & $2$ & $3$ & $c \geq \frac{1}{4} $ &
$ \geq \frac{1}{4} > \frac{5}{32} $ \\ \hline
\ref{f15}\,(4c,4d) & $2$ & $3$ & $c \geq \frac{3}{16}$ &
$ \geq \frac{3}{16} > \frac{5}{32}$ \\ \hline
\ref{f16}\,(a,b) & $3$ & $3$ & $c \geq \frac{9}{64} + \frac{1}{64}$ &
$ \geq \frac{5}{32}$ \\ \hline
\ref{f16}\,(c) & $3$ & $3$ & $c \geq \frac{3}{16}$ &
$ \geq \frac{3}{16} > \frac{5}{32}$ \\ \hline
\hline
\end{tabular}
\caption {Direct proofs for Case 2c and Case 2d.}
\label {tab:case2c2d}
\end {table}

This completes the case analysis and thereby the proof of Theorem~\ref{thm:5/32}.
\qed


\begin{thebibliography}{99}

\bibitem{ACK+11}
A. K.~Abu-Affash, P.~Carmi, M.~J. Katz, and G.~Morgenstern,
Multi cover of a polygon minimizing the sum of areas,
\emph{Int. J. Comput. Geometry \& Appl.} {\bf 21(6)} (2011), 685--698.

\bibitem{AW13}
A.~Adamaszek and A.~Wiese,
Approximation schemes for maximum weight independent set of rectangles,
in \emph{Proc. 54th Sympos. on Found. of Comp. Sci.}, IEEE, 2013.

\bibitem{AW15}
A.~Adamaszek and A.~Wiese,
A quasi-PTAS for the two-dimensional geometric knapsack problem,
in \emph{Proc. 26th ACM-SIAM Sympos. Discrete Algor.}, SIAM, 2015.

\bibitem{AM06}
P.~K.~Agarwal and N.~H.~Mustafa,
Independent set of intersection graphs of convex objects in 2D,
\emph{Comput. Geom.} {\bf 34(2)} (2006), 83--95.

\bibitem{AS87}
A. Aggarwal and S. Suri,
Fast algorithms for computing the largest empty rectangle,
in: \emph{Proc. 3rd Ann. Sympos. Comput. Geometry},
ACM, 1987, pp.~278--290.

\bibitem{Aj73}
M.~Ajtai,
The solution of a problem of T. Rado,
\emph{Bulletin de l'Acad\'emie Polonaise des Sciences,
S\'erie des Sciences Math\'ematiques, Astronomiques et Physiques}
\textbf{21} (1973), 61--63.

\bibitem{A05}
S. Aluru, Quadtrees and octrees, Ch.~19 in
\emph{Handbook of Data Structures and Applications
  (D. P. Mehta and S. Sahni, editors)},
Chapman \& Hall/CRC, 2005.

\bibitem{BT15}
K.~Balas and Cs.~D. T\'oth,
On the number of anchored rectangle packings for a planar point set,
in \emph{Proc. 21st Ann. Internat. Comput. Combin. Conf.},
Springer, 2015, pp.~377--389.

\bibitem{BK14}
N. Bansal and A. Khan,
Improved approximation algorithm for two-dimensional bin packing,
in \emph{Proc. 25th ACM-SIAM Sympos. Discrete Algor.}, SIAM, 2014, pp.~13--25.

\bibitem{BDJ10a}
S.~Bereg, A.~Dumitrescu, and M.~Jiang,
Maximum area independent set in disk intersection graphs,
\emph{Internat. J. Comput. Geom. Appl.}
\textbf{20} (2010), 105--118.

\bibitem{BDJ10b}
S. Bereg, A. Dumitrescu and M. Jiang,
On covering problems of Rado,
\emph{Algorithmica}
\textbf{57} (2010), 538--561.

\bibitem{BVX13}
S. Bhowmick, K. Varadarajan, and S.-K. Xue,
A constant-factor approximation for multi-covering with disks,
in \emph{Proc. 29th ACM Sympos. on Comput. Geom.}, ACM, 2013,  pp.~243--248.

\bibitem{BVX13b}
S. Bhowmick, K. Varadarajan, and S.-K. Xue,
Addendum to ``A constant-factor approximation for multi-covering with disks'',
manuscript, 2014, available
at~\url{http://homepage.cs.uiowa.edu/~kvaradar/papers.html}.

\bibitem{BCK+05}
V. B\'{\i}lo, I. Caragiannis, C. Kaklamanis, and P. Kanellopoulos,
Geometric clustering to minimize the sum of cluster sizes,
in \emph{Proc. European Sympos. Algor.}, LNCS~3669, 2005, pp.~460--471.

\bibitem{CC09}
P. Chalermsook and J. Chuzhoy,
Maximum independent set of rectangles,
in \emph{Proc. 20th ACM-SIAM Sympos. Discrete Algor.}, SIAM, 2009, pp.~892--901.

\bibitem{Cha03}
T.~M. Chan,
Polynomial-time approximation schemes for packing and piercing fat objects,
\emph{J. Algorithms} {\bf 46} (2003), 178--189.

\bibitem{CH12}
T. M. Chan and S. Har-Peled,
Approximation algorithms for maximum independent set of pseudodisks,
\emph{Discrete Comput. Geom.} {\bf 48(2)} (2012), 373--392.

\bibitem{C83}
K. L. Clarkson, Fast algorithms for the all nearest neighbors problem,
in \emph{Proc. 24th Annu. IEEE Sympos. Found. Comput. Sci.}, IEEE, 1983, pp.~226--232

\bibitem{CDL86}
B. Chazelle, R.L. Drysdale, and D.T. Lee,
Computing the largest empty rectangle,
\emph{SIAM J. Comput.} {\bf 15(1)} (1986), 300--315.

\bibitem{DJ13a}
A.~Dumitrescu and M.~Jiang,
On the largest empty axis-parallel box amidst $n$ points,
\emph{Algorithmica} {\bf  66(2)} (2013), 225--248.

\bibitem{DJ13b}
A.~Dumitrescu and M.~Jiang,
Computational Geometry Column 56,
\emph{SIGACT News Bulletin}
\textbf{44(2)} (2013), 80--87.

\bibitem{DT15}
A.~Dumitrescu and  Cs.~D. T\'oth,
Packing anchored rectangles,
\emph{Combinatorica} {\bf 35(1)} (2015), 39--61.

\bibitem{EJS01}
T.~Erlebach, K.~Jansen, and E.~Seidel,
Polynomial-time approximation schemes for geometric intersection graphs,
\emph{SIAM J. Comput.} {\bf 34(6)} (2006), 1302--1323.

\bibitem{FW91}
M.~Formann and F. Wagner,
A packing problem with applications to lettering of maps,
in \emph{Proc. 7th Sympos. Comput. Geometry}, 1991, ACM Press, pp.~281--288.

\bibitem{Ha66}
M.~Hanan,
On Steiner's problem with rectilinear distance,
{\em SIAM J. Appl. Math.} \textbf{14} (1966), 255--265.

\bibitem{IL03}
C.~Iturriaga and A.~Lubiw,
Elastic labels around the perimeter of a map,
\emph{J. Algorithms} \textbf{47(1)} (2003), 14--39.

\bibitem{JC04}
J.-W. Jung and K.-Y. Chwa,
Labeling points with given rectangles,
\emph{Inf. Process. Lett.} \textbf{89(3)} (2004), 115--121.

\bibitem{KT13}
K.~G.~Kakoulis and I.~G.~Tollis, Labeling algorithms,
Ch.~28 in \emph{Handbook of Graph Drawing and Visualization (R.~Tamassia, ed.)},
CRC Press, 2013.

\bibitem{KR92}
D. Knuth and A. Raghunathan, The problem of compatible representatives,
\emph{SIAM J. Disc. Math.} \textbf{5} (1992), 36--47.

\bibitem{KNN+02}
A.~Koike, S.-I. Nakano, T.~Nishizeki, T.~Tokuyama, and S.~Watanabe,
Labeling points with rectangles of various shapes,
\emph{Int. J. Comput. Geometry Appl.}
\textbf{12(6)} (2002), 511--528.

\bibitem{KSW99}
M. van Kreveld, T. Strijk, and A. Wolff, Point labeling with sliding labels,
\emph{Comput. Geom.} {\bf 13} (1999), 21--47.

\bibitem{Ra49}
R.~Rado,
Some covering theorems (I),
\emph{Proc. Lond. Math. Soc.}
\textbf{51} (1949), 232--264.

\bibitem{Ra51}
R.~Rado,
Some covering theorems (II),
\emph{Proc. Lond. Math. Soc.}
\textbf{53} (1951), 243--267.

\bibitem{Ra68}
R.~Rado,
Some covering theorems (III),
\emph{Proc. Lond. Math. Soc.}
\textbf{42} (1968), 127--130.

\bibitem{Ra28}
T.~Rado,
Sur un probl\`eme relatif \`a un th\'eor\`eme de Vitali,
\emph{Fund.\ Math.},
\textbf{11} (1928), 228--229.

\bibitem{Tu69}
  W.~Tutte,
  {\it Recent Progress in Combinatorics:
    Proc. 3rd Waterloo Conf. Combin.},
  May 1968, Academic Press, New York, 1969.

\bibitem{Wi07}
P.~Winkler, Packing rectangles, in {\it Mathematical Mind-Benders},
A.K. Peters Ltd., Wellesley, MA, 2007, pp.~133--134.

\bibitem{ZJ06}
B. Zhu and M. Jiang,
A combinatorial theorem on labeling squares with points and its application,
\emph{J. Comb. Optim.}
{\bf 11(4)} (2006), 411--420.

\end{thebibliography}
\end{document}